\documentclass[12pt]{elsarticle}
\usepackage[margin=2.5cm]{geometry}
\usepackage{lipsum}
\usepackage{float}
\usepackage{booktabs}

\makeatletter
\def\ps@pprintTitle{%
   \let\@oddhead\@empty
   \let\@evenhead\@empty
   \def\@oddfoot{\reset@font\hfil\thepage\hfil}
   \let\@evenfoot\@oddfoot
}
\makeatother

\usepackage{amsfonts,color,morefloats,pslatex}
\usepackage{amssymb,amsthm, amsmath,latexsym}
\usepackage{array}

\theoremstyle{thmstyleone}%

\theoremstyle{thmstyletwo}%

\theoremstyle{thmstylethree}%

\usepackage{array}
\usepackage{enumerate}

\newtheorem{theorem}{Theorem}
\newtheorem{lemma}[theorem]{Lemma}

\newtheorem{definition}[theorem]{Definition}
\newtheorem{example}[theorem]{Example}

\newcommand{\ord}{{\mathrm{ord}}}

\usepackage{lineno}
\usepackage{blindtext}
\begin{document}

\begin{frontmatter}



\title{A family of self-orthogonal divisible  codes with locality 2}

\tnotetext[fn1]{
This research was supported in part by the National Natural
Science Foundation of China under Grant 12271059,  and in part
by the Fundamental Research Funds for the Central Universities,
CHD, under Grant 300102122202.
}

\author[1]{Ziling Heng}
\ead{zilingheng@chd.edu.cn}
\author[1]{Mengjie Yang}
\ead{yang199902042023@163.com}
\author[2]{Yang Ming}
\ead{yangming@chd.edu.cn}

\cortext[cor]{Corresponding author}
\address[1]{School of Science, Chang'an University, Xi'an 710064, China}
\address[2]{School of Information Engineering, Chang'an University, Xi'an 710064, China}




\begin{abstract}
Linear codes are widely studied due to their applications in communication, cryptography, quantum codes, distributed storage and many other fields.
In this paper, we use the trace and norm functions over finite fields to construct a  family of linear codes.
The weight distributions of the codes are determined  in three cases via Gaussian sums. The codes are shown to be self-orthogonal divisible codes with only
three, four or five nonzero weights in these cases. In particular, we prove that this family of linear codes has locality 2.  Several optimal or almost optimal linear codes and locally recoverable codes are derived. In particular, an infinite family of distance-optimal binary linear codes with respect to the sphere-packing bound is obtained.
The self-orthogonal codes derived in this paper can be used to construct lattices and have nice application in distributed storage.
\end{abstract}


\end{frontmatter}
\section{Introduction}\label{sec1}
Let $\mathbb{F}_{q}$ denote the finite field with $q$ elements, where $q$ is a power of a prime $p$. An $[n,k,d]$ linear code $\mathcal{C}$ over $\mathbb{F}_{q}$ is a $k$-dimensional subspace of $\mathbb{F}_{q}^{n}$ with minimum Hamming distance $d$.
Let $A_{i}$ denote the number of codewords with Hamming weight $i$ in a code $\mathcal{C}$ with length $n$. The weight enumerator of an $[n,k]$ linear code $\mathcal{C}$ is defined by the polynomial
\begin{eqnarray*}
1+A_{1}z+\cdots+A_{n}z^{n}.
\end{eqnarray*}
The sequence $(1, A_{1}, \cdots, A_{n})$ is called the weight distribution of $\mathcal{C}$. A linear code $\mathcal{C}$ is referred to as a $t$-weight code if the number of nonzero $A_{j}$, $1 \leq j \leq n$, in the sequence $(1, A_{1}, \cdots, A_{n})$ equals $t$.
It is well known that the weight enumerator of a code not only contains its error detection and error correction capabilities, but also plays an important role in calculating the error probability of error detection and correction of the code. The reader is referred to \cite{LR, C, KC, CQ} for known weight enumerators of some linear codes.
For any $[n,k,d]$ linear code, some trade-offs exist among the parameters $n,k,d$.
The Griesmer bound \cite{GJH} for an $[n,k,d]$ linear code over $\mathbb{F}_{q}$ is given by
\begin{eqnarray*}
n\geq\sum_{i=0}^{k-1}\left\lceil\frac{d}{q^{i}}\right\rceil,
\end{eqnarray*}
where $\lceil\cdot\rceil$ is the ceiling function. A linear code meeting the bound with equality is said to be a Griesmer code.
An $[n,k,d]$ linear code over $\mathbb{F}_q$ is said to be distance-optimal if there exists no $[n,k,d+1]$ linear code over $\mathbb{F}_q$.

The divisibility is an important property of linear codes. A linear code over $\mathbb{F}_q$ is said to be divisible if all its codewords have weights divisible by an integer $\Delta>1$. Then the code is said to be $\Delta$-divisible and $\Delta$ is called a divisor of the code.
If $\gcd(\Delta,q)=1$, then $\mathcal{C}$ is equivalent to a code obtained by taking a linear code over $\mathbb{F}_q$, repeating each coordinate $\Delta$ times, and adding on enough $0$ entries to make a code whose length is that of $\mathcal{C}$ \cite{Ward1}.
Hence, the most interesting case is that $\Delta$ is a power of the characteristic of $\mathbb{F}_q$.
Ward introduced the divisible codes in 1981 \cite{Ward1} and gave a survey in 2001 \cite{Ward2}.
Divisible codes have many applications including Galois geometries, subspace codes, partial spreads, vector space partitions, and Griesmer
codes \cite{KK, K, Ward2, Ward3}.

The self-orthogonality is another important property of linear codes.
The dual code of an $[n, k]$ linear code $\mathcal{C}$ is defined as $\mathcal{C}^{\bot}=\{\mathbf{u}\in\mathbb{F}_{q}^{n}:\langle\mathbf{u},\mathbf{v}\rangle=0$ for all $\mathbf{v}\in\mathcal{C}\}$, where $\langle\cdot\rangle$ denotes the standard inner product. Then $\mathcal{C}^{\bot}$ is an $[n, n-k]$ linear code. If a linear code $\mathcal{C}$ satisfies $\mathcal{C}\subseteq\mathcal{C}^{\bot}$, then it is called a \emph{self-orthogonal} code. If a linear code $\mathcal{C}$ satisfies $\mathcal{C}=\mathcal{C}^{\bot}$, then it is said to be self-dual. For a cyclic code $\mathcal{C}$ with generator polynomial $g(x)$, $\mathcal{C}$ is self-orthogonal if and only if $g^{\bot}(x)\mid g(x)$, where $g^{\bot}(x)=x^{k}h(x^{-1})/h(0)$ with $h(x):=\frac{x^{n}-1}{g(x)}$ \cite{WV}. For binary and ternary linear codes, we can judge the self-orthogonality of them by the divisibility of their weights \cite{WV}. However, it is very difficult to judge whether a general $q$-ary linear code is self-orthogonal or not. The reader is referred to \cite{DLM, WL, WL2, ZXCC} for some known self-orthogonal codes.
Recently, in \cite{LH}, Li and Heng gave a useful sufficient condition for a $q$-ary linear code containing the all-1 vector to be self-orthogonal for odd prime power $q$.

\begin{lemma}\label{th-selforthogonal}\cite{LH}
Let $q=p^s$, where $p$ is an odd prime. Let $\mathcal{C}$ be an $[n,k,d]$ linear code over $\mathbb{F}_q$ with $\mathbf{1}\in\mathcal{C}$, where $\mathbf{1}$ is all-$1$ vector of length $n$. If the linear code $\mathcal{C}$ is $p$-divisible, then $\mathcal{C}$ is self-orthogonal.
\end{lemma}

The locality property of a linear code is of interest in distributed data storage systems  \cite{GH, P}.
A linear code is called a locally recoverable code (LRC for short) with locality $r$ if any symbol in the encoding is a function of $r$ other symbols. It means that any symbol can be repaired from at most $r$ other code symbols.
The repair cost of an $[n,k,d]$ LRC with locality $r\ll k$ is lower in comparison to MDS codes. Locally recoverable codes have been implemented in practice by Microsoft and Facebook \cite{CC}, \cite{M}.
In the literature, there are many interesting results on LRCs. In \cite{GH},  Gopalan et al. gave the  Singleton-like bound on the parameters of LRCs.
In \cite{HX}, Hao et al. studied bounds and constructions of LRCs from the viewpoint of parity-check matrices.
In \cite{TB}, Tamo and Barg constructed optimal LRCs of length $n$ by evaluations of special polynomials over finite field $\mathbb{F}_q$ such that $q\geq n$.
After that,  some good polynomials to obtain Tamo-Barg codes with more  flexible parameters were proposed in \cite{LMC, M}.
In \cite{KWG}, some constructions of optimal LRCs with super-linear length were given.
Recently, Tan et al. established some general theory on the minimum locality of linear codes and studied the minimum locality of some special families of linear codes in \cite{TFDT}. Considering the alphabet size of an $[n,k,d]$ LRC over $\mathbb{F}_{q}$ with locality $r$, Cadambe and Mazumdar \cite{CM} presented a  bound as
\begin{eqnarray}\label{CMe}
k\leq\min\limits_{1\leq t\leq n/(r+1)}\left\{tr+k_{\text{opt}}^{(q)}(n-t(r+1),d)\right\},
\end{eqnarray}
where $k_{\text{opt}}^{(q)}(n,d)$ denotes the maximum dimension of a linear code over $\mathbb{F}_{q}$ of length $n$ and minimum distance $d$. An LRC attaining the bound is said to be optimal. An LRC of dimension which equals the bound minus $1$ is said to be almost optimal.

Constructing new linear codes with desirable properties has been an interesting research topic in coding theory.
In this paper, we will present a construction of self-orthogonal divisible linear codes with locality 2 by the trace and norm functions over finite fields.
Let $m$, $m_{1}$, $m_{2}$ be positive integers such that $m_{1}\mid m$, $m_{2}\mid m$ and $\gcd{(m_{1},m_{2})}=e$. Let $\mathrm{Tr}_{q^{m_{i}}/q}$ be the trace function from $\mathbb{F}_{q^{m_{i}}}$ to $\mathbb{F}_{q}$ defined by
$$\mathrm{Tr}_{q^{m_{i}}/q}(x)=\sum_{j=0}^{i-1}x^{q^i},\ x\in \mathbb{F}_{q^{m_{i}}},  i=1,2.$$
 Let $\mathrm{N}_{q^{m}/q^{m_{i}}}$ be the norm function from $\mathbb{F}_{q^{m}}$ to $\mathbb{F}_{q^{m_{i}}}$ defined by
\begin{eqnarray*}
\mathrm{N}_{q^{m}/q^{m_{i}}}(x)=x^{1+q^{m_{i}}+\cdots+(q^{m_{i}})^{\frac{m}{m_{i}}-1}}=x^{\frac{q^{m}-1}{q^{m_{i}}-1}},\ x\in\mathbb{F}_{q^{m}}, i=1,2.
\end{eqnarray*}
Choose a defining set
 \begin{eqnarray}\label{D}
D=\left\{x\in\mathbb{F}_{q^{m}}:\mathrm{Tr}_{q^{m_{2}}/q}(\mathrm{N}_{q^{m}/q^{m_{2}}}(x))=0\right\}.
\end{eqnarray}
Define a $q$-ary linear code by
\begin{eqnarray*}
\mathcal{C}_{D}=\left\{\left(\mathrm{Tr}_{q^{m_{1}}/q}(b\mathrm{N}_{q^{m}/q^{m_{1}}}(x))
\right)_{x\in D}:b\in\mathbb{F}_{q^{m_{1}}}\right\}.
\end{eqnarray*}
Let $D':=D\backslash \{0\}$. We remark that the corresponding linear code $\mathcal{C}_{D'}$ was studied for some special cases in \cite{ZQ}.
The augmented code of $\mathcal{C}_{D}$ is defined by
\begin{eqnarray*}
\overline{\mathcal{C}_{D}}=\left\{\left(\mathrm{Tr}_{q^{m_{1}}/q}(b\mathrm{N}_{q^{m}/q^{m_{1}}}(x)\right)_{x\in D}+c\mathbf{1}: b\in\mathbb{F}_{q^{m_{1}}},c\in\mathbb{F}_{q}\right\},
\end{eqnarray*}
where $\mathbf{1}$ is the all-1 vector of length $\vert D\vert$. As was pointed out by Ding and Tang in \cite{DT}, it is usually a challenge to determine the parameters of the augmented code of a linear code as we may require the complete weight information of this linear code.
In \cite{HXYQ}, the parameters and weight distribution of $\overline{\mathcal{C}_{D}}$ were given for the case $m=m_1=2m_2$.
In this paper, we will mainly prove that $\overline{\mathcal{C}_{D}}$ has locality 2 and determine the parameters and weight distribution of $\overline{\mathcal{C}_{D}}$ in three more cases. Besides, we will prove that $\overline{\mathcal{C}_{D}}$ is self-orthogonal and $q$-divisible in these cases.
In particular, many optimal or almost optimal linear codes and locally recoverable codes will also be derived.

\section{Preliminaries}\label{sec2}

In this section, we present some notations and basic results on additive and multiplicative characters of finite fields and Gaussian sums.
\subsection{Notations}
The following notations will be used in this paper:

$\chi$, $\chi_{1}$, $\chi_{2}$: ~canonical additive characters of $\mathbb{F}_{q}$, $\mathbb{F}_{q^{m_{1}}}$, $\mathbb{F}_{q^{m_{2}}}$, respectively;

$\lambda$, $\lambda_{1}$, $\lambda_{2}$: ~generators of multiplicative character groups of $\mathbb{F}_{q}$, $\mathbb{F}_{q^{m_{1}}}$, $\mathbb{F}_{q^{m_{2}}}$, respectively;

$G(\lambda)$, $G(\lambda_{1})$, $G(\lambda_{2})$: ~Gaussian sums over $\mathbb{F}_{q}$, $\mathbb{F}_{q^{m_{1}}}$,
$\mathbb{F}_{q^{m_{2}}}$, respectively;

$\alpha$ ,$\alpha_{1}$, $\alpha_{2}$: ~primitive elements of $\mathbb{F}_{q^{m}}^{*}$, $\mathbb{F}_{q^{m_{1}}}^{*}$, $\mathbb{F}_{q^{m_{2}}}^{*}$, respectively;

$e$: ~$e=\gcd(m_{1},m_{2})$;

$l$: ~$l=\gcd(\frac{m_{2}}{e},q-1)$.
\subsection{Characters of finite fields and Gaussian sums}
Let $q=p^{s}$ with $p$ a prime. Denote the primitive $p$-th root of complex unity by $\zeta_{p}$. Define the additive character of $\mathbb{F}_{q}$ by the homomorphism $\chi$ from the additive group $\mathbb{F}_{q}$ to the complex multiplicative group $\mathbb{C}^{*}$
such that $\chi(x+y)=\chi(x)\chi(y)$ for all $x, y\in\mathbb{F}_{q}$.
For any $a\in\mathbb{F}_{q}$, an additive character of $\mathbb{F}_{q}$ can be defined by the function $\chi_{a}(x)=\zeta_{p}^{\mathrm{Tr}_{q/p}(ax)}, x\in\mathbb{F}_{q}$, where $\mathrm{Tr}_{q/p}(x)$ is the trace function from $\mathbb{F}_{q}$ to $\mathbb{F}_{p}$. In addition, the set $\widehat{\mathbb{F}_{q}}:=\{\chi_{a}:a\in\mathbb{F}_{q}\}$ is called the \emph{additive character group} of $\mathbb{F}_q$ which consists of all $q$ different additive characters of $\mathbb{F}_{q}$. By definition, it is obvious that $\chi_{a}(x) = \chi_{1}(ax)$. In particular, $\chi_{0}$ is referred to as the trivial additive character  and $\chi_{1}$ is called the canonical additive character of $\mathbb{F}_{q}$. The orthogonal relation of additive characters (see\cite{RH}) is given by
\begin{eqnarray*}
\sum\limits_{x\in\mathbb{F}_{q}}\chi_{a}(x)=
\begin{cases}
q &  \mbox{ if }a=0,\\
0 &  \mbox{ otherwise}.
\end{cases}
\end{eqnarray*}

Let $\mathbb{F}_{q}^{*}=\langle\beta\rangle$. For each $0\leq j\leq q-2$, a multiplicative character of $\mathbb{F}_{q}$ is defined as the homomorphism $\psi$ from the multiplicative group $\mathbb{F}_{q}^{*}$ to the complex multiplicative group $\mathbb{C}^{*}$ such that $\psi(xy)=\psi(x)\psi(y)$ for all $x, y\in \mathbb{F}_{q}^{*}$. The function $\psi_{j}(\beta^{k})=\zeta_{q-1}^{jk}$ for $k=0, 1, \cdots, q-2$ gives a multiplicative character, where $0\leq j\leq q-2$. The set $\widehat{\mathbb{F}_{q}^{*}}:=\{\psi_{j}: j = 0, 1, \cdots, q-2\}$ is called the \emph{multiplicative character group} of $\mathbb{F}_q$
which consists of all the multiplicative characters of $\mathbb{F}_{q}$. In particular, $\psi_{0}$ is called the trivial multiplicative character and $\eta:=\psi_{\frac{q-1}{2}}$ for odd $q$ is referred to as the quadratic multiplicative character of $\mathbb{F}_{q}$. The orthogonal relation of multiplicative characters (see\cite{RH}) is given by
\begin{eqnarray*}
\sum\limits_{x\in\mathbb{F}_{q}^{*}}\psi_{j}(x)=
\begin{cases}
q-1 &  \mbox{if}~j=0,\\
0 &  \mbox{if}~j\neq0.
\end{cases}
\end{eqnarray*}
The conjugate $\bar{\psi}$ of a multiplicative character of $\mathbb{F}_q$ is defined by $\bar{\psi}(x)=\overline{\psi(x)}$ for $x\in \mathbb{F}_q^*$.

For an additive character $\chi$ and a multiplicative character $\psi$ of $\mathbb{F}_{q}$, the Gaussian sum over $\mathbb{F}_{q}$ is defined by
\begin{eqnarray*}
G(\psi,\chi)=\sum\limits_{x\in\mathbb{F}_{q}^{*}}\psi(x)\chi(x).
\end{eqnarray*}
In the following, we use $G(\psi)$ to denote $G(\psi,\chi)$ for simplicity if $\chi$ is the canonical additive character.
We can easily prove $G(\psi_{0})=-1$ and $G(\overline{\psi}) = \psi(-1)\overline{G(\psi)}$. If $\psi\neq\psi_{0}$, we have $\vert G(\psi)\vert =\sqrt{q}$ \cite{RH}. Gaussian sums can be viewed as the Fourier coefficients in the Fourier expansion of the restriction of $\chi$ to $\mathbb{F}_{q}^{*}$ in terms of the multiplicative characters of $\mathbb{F}_{q}$ \cite{RH},
i.e.
\begin{eqnarray}\label{s1}
\chi(x)=\frac{1}{q-1}\sum\limits_{\psi\in\widehat{\mathbb{F}}_{q}^{*}}G(\overline{\psi})\psi(x),x\in\mathbb{F}_{q}^{*},
\end{eqnarray}
where $\chi$ is the canonical additive character.
In this paper, Gaussian sum is an vital tool to compute exponential sums. In general, determining the explicit values of Gauss sums is difficult. In some special cases, Gaussian sums were explicitly determined in \cite{BRK}\cite{RH}.

In the following, we state the Gaussian sums in the semi-primitive case.
\begin{lemma}\label{le2}\cite[Semi-primitive case Gaussian sums]{BRK}
Let $\phi$ be a multiplicative character of order $N$ of $\mathbb{F}_{r}^{*}$. Assume that $N\neq 2$ and there exists a least positive integer $j$ such that $p^{j}\equiv -1\pmod N$. Let $r=p^{2j\gamma}$ for some integer $\gamma$. Then the Gaussian sums of order $N$ over $\mathbb{F}_{r}$ are given by
\begin{eqnarray*}
G(\phi)=
\begin{cases}
(-1)^{\gamma-1}\sqrt{r} & if ~p=2,\\
(-1)^{\gamma-1+\frac{\gamma(p^{j}+1)}{N}}\sqrt{r} & if ~p\geq3.
\end{cases}
\end{eqnarray*}
Furthermore, for $1\leq s\leq N-1$, the Gaussian sums $G(\phi^{s})$ are given by
\begin{eqnarray*}
G(\phi^{s})=
\begin{cases}
(-1)^{s}\sqrt{r} & if ~N ~is ~even, ~p, ~\gamma ~and ~\frac{p^{j}+1}{N} ~are ~odd,\\
(-1)^{\gamma-1}\sqrt{r} & otherwise.
\end{cases}
\end{eqnarray*}
\end{lemma}

The well-known quadratic Gaussian sums are determined in the following.
\begin{lemma}\label{le4}\cite[Theorem 5.15]{RH}
Suppose that $q=p^{s}$ and $\eta$ is the quadratic multiplicative character of $\mathbb{F}_{q}$, where $p$ is an odd prime. Then
\begin{eqnarray*}
G(\eta)=(-1)^{s-1}(\sqrt{p^{*}})^{s}=
\begin{cases}
(-1)^{s-1}\sqrt{q} & if ~p\equiv1\pmod4,\\
(-1)^{s-1}(\sqrt{-1})^{s}\sqrt{q} & if ~p\equiv3\pmod4,
\end{cases}
\end{eqnarray*}
where $p^{*}=(-1)^{\frac{p-1}{2}}p$.
\end{lemma}

\section{Some exponential sums}\label{sec3}

In this section, we investigate two exponential sums which will be used to calculate the weight distribution of $\overline{\mathcal{C}_{D}}$.

Let $m$, $m_{1}$, $m_{2}$ be positive integers such that $m_{1}\mid m$, $m_{2}\mid m$. Let $\chi$ be the canonical additive character of $\mathbb{F}_{q}$. Let $\chi_{i}$ be the canonical additive characters of $\mathbb{F}_{q^{m_{i}}} , i = 1, 2$, respectively.
Let $b\in \mathbb{F}_{q^{m_1}}^*$ and $c\in \mathbb{F}_q^*$.
Denote
\begin{eqnarray*}
\Delta(b)=\sum\limits_{x\in\mathbb{F}_{q^{m}}^{*}}\sum\limits_{y,z\in\mathbb{F}_{q}^{*}}\chi_{1}(yb\mathrm{N}_{q^{m}/q^{m_{1}}}(x))\chi_{2}(z\mathrm{N}_{q^{m}/q^{m_{2}}}(x)),
\end{eqnarray*}
and
\begin{eqnarray*}
\Omega(b,c)=\sum\limits_{x\in\mathbb{F}_{q^{m}}^{*}}\sum\limits_{y,z\in\mathbb{F}_{q}^{*}}
\chi_{1}(yb\mathrm{N}_{q^{m}/q^{m_{1}}}(x))\chi(yc)\chi_{2}(z\mathrm{N}_{q^{m}/q^{m_{2}}}(x)).
\end{eqnarray*}

\begin{lemma}\label{db}\cite [Lemma 3.5]{ZQ}
Let $m$, $m_{1}$, $m_{2}$ be positive integers such that $m_{1}\mid m$, $m_{2}\mid m$. Denote $\gcd(m_{1}$, $m_{2})=e$. Let $t_{i} =\frac{q^{m_{i}}-1}{q-1} , i = 1, 2$. The value distribution of $\Delta(b), b\in\mathbb{F}_{q^{m_{1}}}^{*}$ is given as follows.\\
\begin{enumerate}[(1)]
\item If $e=1$, then $\Delta(b)=\frac{(q^{m}-1)(q-1)^{2}}{(q^{m_{1}}-1)(q^{m_{2}}-1)}$ for all $b\in\mathbb{F}_{q^{m_{1}}}^{*}$.\\
\item If $e=2$, then
\begin{eqnarray*}
\Delta(b)=
\begin{cases}
\frac{(q^{m}-1)(q-1)^{2}}{(q^{m_{1}}-1)(q^{m_{2}}-1)}\left(1+(-1)^{\frac{m_{1}+m_{2}}{2}}q^{\frac{m_{1}+m_{2}}{2}+1}\right) &  \frac{q^{m_{1}}-1}{q+1} ~times,\\
\frac{(q^{m}-1)(q-1)^{2}}{(q^{m_{1}}-1)(q^{m_{2}}-1)}\left(1+(-1)^{\frac{m_{1}+m_{2}}{2}+1}q^{\frac{m_{1}+m_{2}}{2}}\right) &  \frac{q(q^{m_{1}}-1)}{q+1} ~times.
\end{cases}
\end{eqnarray*}
\end{enumerate}
\end{lemma}

In the following, we investigate the exponential sum $\Omega(b,c), b\in\mathbb{F}_{q^{m_{1}}}^{*}, c\in\mathbb{F}_{q}^{*}$.

\begin{lemma}\label{le3}
Let $m$, $m_{1}$, $m_{2}$ be positive integers such that $m_{1}\mid m$, $m_{2}\mid m$. Denote $\gcd(m_{1},m_{2})=e$, $\gcd(\frac{m_{2}}{e},q-1)=l$. Let $\widehat{\mathbb{F}}_{q}^{*}=\langle\lambda\rangle$, $\widehat{\mathbb{F}}_{q^{m_{i}}}^{*}=\langle\lambda_{i}\rangle$, $t_{i} =\frac{q^{m_{i}}-1}{q-1} , i = 1, 2$. For $( b,c)\in \mathbb{F}_{q^{m_{1}}}^{*}\times\mathbb{F}_{q}^{*}$, then we have
\begin{eqnarray*}
\Omega(b,c)=\frac{(q^{m}-1)(q-1)}{(q^{m_{1}}-1)(q^{m_{2}}-1)}\sum\limits_{s\in S}G(\lambda_{1}^{t_{1}s})G(\overline{\lambda}_{2}^{t_{2}s})\overline{\lambda}_{1}^{t_{1}s}(b)G(\overline{\lambda}^{\frac{m_{1}}{e}s})\lambda_{1}^{t_{1}s}(c),
\end{eqnarray*}
where $S=\left\{\frac{q-1}{l}j:j=0,1,\cdots,\frac{q^{e}-1}{q-1}l-1\right\}$.
\end{lemma}
\begin{proof}
For $\mathbb{F}_{q^{m}}^{*}=\langle\alpha\rangle$, let $\alpha_{1}=\alpha^{\frac{q^{m}-1}{q^{m_{1}-1}}}$ and $\alpha_{2}=\alpha^{\frac{q^{m}-1}{q^{m_{2}-1}}}$. Then $\mathbb{F}_{q^{m_{1}}}^{*}=\langle\alpha_{1}\rangle$ and $\mathbb{F}_{q^{m_{2}}}^{*}=\langle\alpha_{2}\rangle$. We have
\begin{eqnarray*}
\Omega(b,c)&=&\sum\limits_{y,z\in\mathbb{F}_{q}^{*}}\sum_{i=0}^{q^{m}-2}\chi_{1}(yb\alpha^{\frac{q^{m}-1}{q^{m_{1}}-1}i})\chi(yc)\chi_{2}(z\alpha^{\frac{q^{m}-1}{q^{m_{2}}-1}i})\\
&=&\sum\limits_{y,z\in\mathbb{F}_{q}^{*}}\sum_{i=0}^{q^{m}-2}\chi_{1}(yb\alpha_{1}^{i})\chi(yc)\chi_{2}(z\alpha_{2}^{i})
\end{eqnarray*}
By the Fourier expansion of additive characters in Equation (\ref{s1}), we have
\begin{eqnarray*}
\Omega(b,c)&=&\frac{1}{(q^{m_{1}}-1)(q^{m_{2}}-1)}\sum\limits_{y,z\in\mathbb{F}_{q}^{*}}\chi(yc)\sum_{i=0}^{q^{m}-2}\sum\limits_{\psi_{1}\in\widehat{\mathbb{F}}_{q^{m_{1}}}^{*}}G(\overline{\psi}_{1})\psi_{1}(yb\alpha_{1}^{i})\\
&\;&\times\sum\limits_{\psi_{2}\in\widehat{\mathbb{F}}_{q^{m_{2}}}^{*}}G(\overline{\psi}_{2})\psi_{2}(z\alpha_{2}^{i})\\
&=&\frac{1}{(q^{m_{1}}-1)(q^{m_{2}}-1)}\sum\limits_{y,z\in\mathbb{F}_{q}^{*}}\chi(yc)\sum\limits_{\substack{\psi_{j}\in\widehat{\mathbb{F}}_{q^{m_{j}}}^{*}\\j=1,2}}G(\overline{\psi}_{1})G(\overline{\psi}_{2})\psi_{1}(yb)\psi_{2}(z)\\
&\;&\times\sum_{i=0}^{q^{m}-2}\psi_{1}(\alpha_{1}^{i})\psi_{2}(\alpha_{2}^{i}).
\end{eqnarray*}
Since $m_{i}\mid m$, we obtain $\ord(\psi_{i})\mid(q^{m}-1)$, where $i=1,2$. Therefore, we have
\begin{eqnarray*}
(\psi_{1}(\alpha_{1})\psi_{2}(\alpha_{2}))^{q^{m}-1}=1
\end{eqnarray*}
and
\begin{eqnarray*}
\sum_{i=0}^{q^{m}-2}\psi_{1}(\alpha_{1}^{i})\psi_{2}(\alpha_{2}^{i})&=&\sum_{i=0}^{q^{m}-2}(\psi_{1}(\alpha_{1})\psi_{2}(\alpha_{2}))^{i}\\&=&
\begin{cases}
q^{m}-1 & \mbox{if}~\psi_{1}(\alpha_{1})\psi_{2}(\alpha_{2})=1,\\
0 & \mbox{otherwise}.\\
\end{cases}
\end{eqnarray*}
Let $\widehat{\mathbb{F}}_{q^{m_{i}}}^{*}=\langle\lambda_{i}\rangle$ such that $\lambda_{i}(\alpha_{i})=\zeta_{q^{m_{i}}-1}$, where $i=1,2$. Assume that $\psi_{1}=\lambda_{1}^{u}$ and $\psi_{2}=\lambda_{2}^{v}$ for $0\leq u\leq q^{m_{1}}-2$ and $0\leq v\leq q^{m_{2}}-2$. If $\psi_{1}(\alpha_{1})\psi_{2}(\alpha_{2})=1$, then $\zeta_{q^{m_{1}}-1}^{u}\zeta_{q^{m_{2}}-1}^{v}=1$ which is equivalent to
\begin{eqnarray}
(q^{m_{2}}-1)u+(q^{m_{1}}-1)v\equiv0\pmod{(q^{m_{1}}-1)(q^{m_{2}}-1)}.
\end{eqnarray}
This implies that $(q^{m_{2}}-1)u+(q^{m_{1}}-1)v\equiv0\pmod{q^{m_{i}}-1},i=1,2$. Therefore, $(q^{m_{2}}-1)u\equiv0\pmod{q^{m_{1}}-1}$ and $(q^{m_{1}}-1)v\equiv0\pmod{q^{m_{2}}-1}$. It is known that
\begin{eqnarray*}
\gcd(q^{m_{1}}-1,q^{m_{2}}-1)=q^{\gcd(m_{1},m_{2})}-1=q^{e}-1.
\end{eqnarray*}
Then we have $u\equiv0\pmod{\frac{q^{m_{1}}-1}{q^{e}-1}}$ and $v\equiv0\pmod{\frac{q^{m_{2}}-1}{q^{e}-1}}$. Denote $u=t_{1}s_{1}$ and $v=t_{2}s_{2}$ for $0\leq s_{1}$, $s_{2}\leq q^{e}-2$, where $t_{1}=\frac{q^{m_{1}}-1}{q^{e}-1}$, $t_{2}=\frac{q^{m_{2}}-1}{q^{e}-1}$. Substituting $u=t_{1}s_{1}$, $v=t_{2}s_{2}$ into Equation(1), we have $s_{1}+s_{2}=q^{e}-1$. Hence,
\begin{eqnarray*}
\Omega(b,c)&=&\frac{q^{m}-1}{(q^{m_{1}}-1)(q^{m_{2}}-1)}\sum\limits_{y,z\in\mathbb{F}_{q}^{*}}\chi(yc)\sum_{s_{2}=0}^{q^{e}-2}G(\lambda_{1}^{t_{1}s_{2}})G(\overline{\lambda}_{2}^{t_{2}s_{2}})\overline{\lambda}_{1}^{t_{1}s_{2}}(yb)\lambda_{2}^{t_{2}s_{2}}(z)\\
&=&\frac{q^{m}-1}{(q^{m_{1}}-1)(q^{m_{2}}-1)}\sum_{s_{2}=0}^{q^{e}-2}G(\lambda_{1}^{t_{1}s_{2}})G(\overline{\lambda}_{2}^{t_{2}s_{2}})\overline{\lambda}_{1}^{t_{1}s_{2}}(b)\sum\limits_{y\in\mathbb{F}_{q}^{*}}\chi(yc)\overline{\lambda}_{1}^{t_{1}s_{2}}(y)\\
&\;&\times\sum\limits_{z\in\mathbb{F}_{q}^{*}}\lambda_{2}^{t_{2}s_{2}}(z).
\end{eqnarray*}
Assume that $\mathbb{F}_{q}^{*}=\langle\beta\rangle$, where $\beta=\alpha_{1}^{\frac{q^{m_{1}}-1}{q-1}}=\alpha_{2}^{\frac{q^{m_{2}}-1}{q-1}}$. Hence, $\lambda_{1}(\beta)=\lambda_{2}(\beta)=\zeta_{q-1}$.\\
Since $\gcd(\frac{m_{2}}{e},q-1)=l$, we have
\begin{eqnarray*}
\sum\limits_{z\in\mathbb{F}_{q}^{*}}\lambda_{2}^{t_{2}s_{2}}(z)&=&\sum_{i=0}^{q-2}\lambda_{2}^{t_{2}s_{2}}(\beta^{i})=\sum_{i=0}^{q-2}(\zeta_{q-1}^{\frac{m_{2}s_{2}}{e}})^{i}\\&=&
\begin{cases}
q-1 & \mbox{if}~s_{2}\equiv0\pmod{\frac{q-1}{l}},\\
0 & \mbox{otherwise}.\\
\end{cases}
\end{eqnarray*}
Denote $S=\{s_{2}:s_{2}\equiv0\pmod{\frac{q-1}{l}}, 0\leq s_{2}\leq q^{e}-2\}$. Let $\widehat{\mathbb{F}}_{q}^{*}=\langle\lambda\rangle$. Since $\lambda_{1}(z)=\lambda(z)$ for $z\in\mathbb{F}_{q}$, we have that
\begin{eqnarray*}
\Omega(b,c)&=&\frac{(q^{m}-1)(q-1)}{(q^{m_{1}}-1)(q^{m_{2}}-1)}\sum\limits_{s_{2}\in S}G(\lambda_{1}^{t_{1}s_{2}})G(\overline{\lambda}_{2}^{t_{2}s_{2}})\overline{\lambda}_{1}^{t_{1}s_{2}}(b)\sum\limits_{y\in\mathbb{F}_{q}^{*}}\chi(yc)\overline{\lambda}_{1}^{t_{1}s_{2}}(y)\\
&=&\frac{(q^{m}-1)(q-1)}{(q^{m_{1}}-1)(q^{m_{2}}-1)}\sum\limits_{s_{2}\in S}G(\lambda_{1}^{t_{1}s_{2}})G(\overline{\lambda}_{2}^{t_{2}s_{2}})\overline{\lambda}_{1}^{t_{1}s_{2}}(b)\sum\limits_{y\in\mathbb{F}_{q}^{*}}\chi(yc)\overline{\lambda}^{\frac{m_{1}}{e}s_{2}}(y)\\
&=&\frac{(q^{m}-1)(q-1)}{(q^{m_{1}}-1)(q^{m_{2}}-1)}\sum\limits_{s_{2}\in S}G(\lambda_{1}^{t_{1}s_{2}})G(\overline{\lambda}_{2}^{t_{2}s_{2}})\overline{\lambda}_{1}^{t_{1}s_{2}}(b)G(\overline{\lambda}^{\frac{m_{1}}{e}s_{2}})\lambda_{1}^{t_{1}s_{2}}(c).
\end{eqnarray*}
The proof is completed.
\end{proof}

By Lemma \ref{le3}, the value distribution of $\Omega(b,c)$ can be determined if the Gauss sums are known. In the following, we mainly consider some special cases to give the value distribution of $\Omega(b,c)$.

\begin{lemma}\label{ob}
Let $l=1$ and other notations and hypothesises be the same as those of Lemma \ref{le3}. If $(b,c)$ runs through $\mathbb{F}_{q^{m_{1}}}^{*}\times\mathbb{F}_{q}^{*}$, then the value distribution of $\Omega(b,c), b\in \mathbb{F}_{q^{m_{1}}}^{*}, c\in \mathbb{F}_{q}^{*}$ is given as follows.\\
\begin{enumerate}[(1)]
\item If $e=1$, then
\begin{eqnarray*}
\Omega(b,c)=\frac{-(q^{m}-1)(q-1)}{(q^{m_{1}}-1)(q^{m_{2}}-1)}, b\in \mathbb{F}_{q^{m_{1}}}^{*}, c\in \mathbb{F}_{q}^{*}.
\end{eqnarray*}
\item If $e=2$, then
\begin{eqnarray*}
\Omega(b,c)=
\begin{cases}
\frac{-(q^{m}-1)(q-1)}{(q^{m_{1}}-1)(q^{m_{2}}-1)}\left(1+(-1)^{\frac{m_{1}+m_{2}}{2}}q^{\frac{m_{1}+m_{2}}{2}+1}\right) & \frac{(q^{m_{1}}-1)(q-1)}{q+1} ~times,\\
\frac{-(q^{m}-1)(q-1)}{(q^{m_{1}}-1)(q^{m_{2}}-1)}\left(1+(-1)^{\frac{m_{1}+m_{2}}{2}+1}q^{\frac{m_{1}+m_{2}}{2}}\right) & \frac{q(q^{m_{1}}-1)(q-1)}{q+1} ~times.
\end{cases}
\end{eqnarray*}
\end{enumerate}
\end{lemma}
\begin{proof}
If $l=1$, by Lemma \ref{le3}, we have that $G(\overline{\lambda}^{\frac{m_{1}}{e}s_{2}})=-1$ and
\begin{eqnarray*}
\Omega(b,c)=\frac{-(q^{m}-1)(q-1)}{(q^{m_{1}}-1)(q^{m_{2}}-1)}\sum\limits_{s\in S}G(\lambda_{1}^{t_{1}s})G(\overline{\lambda}_{2}^{t_{2}s})\overline{\lambda}_{1}^{t_{1}s}(b)\lambda_{1}^{t_{1}s}(c),
\end{eqnarray*}
where $S=\left\{(q-1)j:j=0, 1, \ldots, \frac{q^{e}-1}{q-1}-1\right\}$. In the following, we discuss the value distribution of the exponential sum
$\Omega(b,c)$ for $e = 1, 2$, respectively.\\
\begin{enumerate}[(1)]
\item Assume that $e=1$. It is clear that $S=\{0\}$. Then
\begin{eqnarray*}
\Omega(b,c)=\frac{-(q^{m}-1)(q-1)}{(q^{m_{1}}-1)(q^{m_{2}}-1)}, b\in \mathbb{F}_{q^{m_{1}}}^{*}, c\in \mathbb{F}_{q}^{*}.
\end{eqnarray*}
\item Assume that $e=2$. Then we have $S=\{(q-1)j:j=0, 1, \cdots, q\}$. Hence,
\end{enumerate}
\begin{eqnarray*}
\Omega(b,c)&=&\frac{-(q^{m}-1)(q-1)}{(q^{m_{1}}-1)(q^{m_{2}}-1)}\sum_{j=0}^{q}G(\lambda_{1}^{\frac{q^{m_{1}}-1}{q+1}j})G(\overline{\lambda}_{2}^{\frac{q^{m_{2}}-1}{q+1}j})\overline{\lambda}_{1}^{\frac{q^{m_{1}}-1}{q+1}j}(b)\lambda_{1}^{\frac{q^{m_{1}}-1}{q+1}j}(c)\\
&=&\frac{-(q^{m}-1)(q-1)}{(q^{m_{1}}-1)(q^{m_{2}}-1)}
\bigg(1+\sum_{j=1}^{q}G(\lambda_{1}^{\frac{q^{m_{1}}-1}{q+1}j})G(\overline{\lambda}_{2}^{\frac{q^{m_{2}}-1}{q+1}j})\overline{\lambda}_{1}^{\frac{q^{m_{1}}-1}{q+1}j}(b)\\
&\;&\times \lambda_{1}^{\frac{q^{m_{1}}-1}{q+1}j}(c)\bigg).
\end{eqnarray*}
Note that $\mathrm{ord}(\lambda_{1}^{\frac{q^{m_{1}}-1}{q+1}})=\mathrm{ord}(\lambda_{2}^{\frac{q^{m_{2}}-1}{q+1}})=q+1$. Now we give the value distribution of $\Omega(b,c)$ in several cases.\\
\begin{itemize}
\item If $q$ is even, by Lemma \ref{le2} we have
\begin{eqnarray*}
G(\lambda_{1}^{\frac{q^{m_{1}}-1}{q+1}j})G(\overline{\lambda}_{2}^{\frac{q^{m_{2}}-1}{q+1}j})=(-1)^{\frac{m_{1}+m_{2}}{2}}q^{\frac{m_{1}+m_{2}}{2}}, 1\leq j\leq q.
\end{eqnarray*}
Then
\begin{eqnarray*}
\Omega(b,c)&=&\frac{-(q^{m}-1)(q-1)}{(q^{m_{1}}-1)(q^{m_{2}}-1)}\Bigg(1+(-1)^{\frac{m_{1}+m_{2}}{2}}q^{\frac{m_{1}+m_{2}}{2}}
\sum_{j=1}^{q}\overline{\lambda}_{1}^{\frac{q^{m_{1}}-1}{q+1}j}(b)\\
&\;&\times \lambda_{1}^{\frac{q^{m_{1}}-1}{q+1}j}(c)\Bigg).
\end{eqnarray*}
Let $b=\alpha_{1}^{s_{1}}, c=\alpha_{1}^{s_{2}}, 0\leq s_{1}\leq q^{m_{1}}-2, 0\leq s_{2}\leq q-2$. Then we have
\begin{eqnarray*}
\sum_{j=1}^{q}\overline{\lambda}_{1}^{\frac{q^{m_{1}}-1}{q+1}j}(b)\lambda_{1}^{\frac{q^{m_{1}}-1}{q+1}j}(c)
&=&\sum_{j=1}^{q}\zeta_{q+1}^{j(s_{2}-s_{1})}\\
&=&\begin{cases}
q & if ~s_{2}-s_{1}\equiv0\pmod{q+1},\\
-1 & otherwise.
\end{cases}
\end{eqnarray*}
Hence, the value distribution of $\Omega(b,c)$ is
\begin{eqnarray*}
\Omega(b,c)=
\begin{cases}
\frac{-(q^{m}-1)(q-1)}{(q^{m_{1}}-1)(q^{m_{2}}-1)}\left(1+(-1)^{\frac{m_{1}+m_{2}}{2}}q^{\frac{m_{1}+m_{2}}{2}+1}\right) & \frac{(q^{m_{1}}-1)(q-1)}{q+1} ~times,\\
\frac{-(q^{m}-1)(q-1)}{(q^{m_{1}}-1)(q^{m_{2}}-1)}\left(1+(-1)^{\frac{m_{1}+m_{2}}{2}+1}q^{\frac{m_{1}+m_{2}}{2}}\right) & \frac{q(q^{m_{1}}-1)(q-1)}{q+1} ~times.
\end{cases}
\end{eqnarray*}
\item If $q$ is odd and $m_{1}\equiv0\pmod4$, we have $m_{2}\equiv2\pmod4$ due to $\gcd(\frac{m_{2}}{2},q-1)=1$. Since $\frac{m_{2}}{2}$ is odd and $\frac{m_{1}}{2}$ is even, by Lemma \ref{le2} we have
\begin{eqnarray*}
G(\lambda_{1}^{\frac{q^{m_{1}}-1}{q+1}j})G(\overline{\lambda}_{2}^{\frac{q^{m_{2}}-1}{q+1}j})=(-1)^{1+j}q^{\frac{m_{1}+m_{2}}{2}}, 1\leq j\leq q.
\end{eqnarray*}
For $b=\alpha_{1}^{s_{1}}, c=\alpha_{1}^{s_{2}}, 0\leq s_{1}\leq q^{m_{1}}-2, 0\leq s_{2}\leq q-2$, we have
\begin{eqnarray*}
\Omega(b,c)&=&\frac{-(q^{m}-1)(q-1)}{(q^{m_{1}}-1)(q^{m_{2}}-1)}\left(1+q^{\frac{m_{1}+m_{2}}{2}}\sum_{j=1}^{q}(-1)^{1+j}\overline{\lambda}_{1}^{\frac{q^{m_{1}}-1}{q+1}j}(b)\right.\\
&\;&\left.\times\lambda_{1}^{\frac{q^{m_{1}}-1}{q+1}j}(c)\right)\\
&=&\frac{-(q^{m}-1)(q-1)}{(q^{m_{1}}-1)(q^{m_{2}}-1)}\left(1+q^{\frac{m_{1}+m_{2}}{2}}\sum_{j=1}^{q}(-1)^{1+j}\zeta_{q+1}^{j(s_{2}-s_{1})}\right).
\end{eqnarray*}
For $s_{2}-s_{1}\equiv0\pmod{q+1}$, we have
\begin{eqnarray*}
\sum_{j=1}^{q}(-1)^{1+j}\zeta_{q+1}^{j(s_{2}-s_{1})}=1
\end{eqnarray*}
and
\begin{eqnarray*}
\Omega(b,c)=\frac{-(q^{m}-1)(q-1)}{(q^{m_{1}}-1)(q^{m_{2}}-1)}\left(1+q^{\frac{m_{1}+m_{2}}{2}}\right).
\end{eqnarray*}
For $s_{2}-s_{1}\equiv\frac{q+1}{2}\pmod{q+1}$, we have
\begin{eqnarray*}
\sum_{j=1}^{q}(-1)^{1+j}\zeta_{q+1}^{j(s_{2}-s_{1})}=-q
\end{eqnarray*}
and
\begin{eqnarray*}
\Omega(b,c)=\frac{-(q^{m}-1)(q-1)}{(q^{m_{1}}-1)(q^{m_{2}}-1)}\left(1-q^{\frac{m_{1}+m_{2}}{2}+1}\right).
\end{eqnarray*}
For $s_{2}-s_{1}\not\equiv0,\frac{q+1}{2}\pmod{q+1}$, one can see that
\begin{eqnarray*}
\zeta_{q+1}^{s_{2}-s_{1}}+\zeta_{q+1}^{3(s_{2}-s_{1})}+\cdots+\zeta_{q+1}^{q(s_{2}-s_{1})}=0
\end{eqnarray*}
and
\begin{eqnarray*}
\zeta_{q+1}^{2(s_{2}-s_{1})}+\zeta_{q+1}^{4(s_{2}-s_{1})}+\cdots+\zeta_{q+1}^{(q-1)(s_{2}-s_{1})}=-1.
\end{eqnarray*}
This implies that
\begin{eqnarray*}
\sum_{j=1}^{q}(-1)^{1+j}\zeta_{q+1}^{j(s_{2}-s_{1})}=1
\end{eqnarray*}
and
\begin{eqnarray*}
\Omega(b,c)=\frac{-(q^{m}-1)(q-1)}{(q^{m_{1}}-1)(q^{m_{2}}-1)}\left(1+q^{\frac{m_{1}+m_{2}}{2}}\right).
\end{eqnarray*}
Hence, the value distribution of $\Omega(b,c)$ is
\begin{eqnarray*}
\Omega(b,c)=
\begin{cases}
\frac{-(q^{m}-1)(q-1)}{(q^{m_{1}}-1)(q^{m_{2}}-1)}\left(1-q^{\frac{m_{1}+m_{2}}{2}+1}\right) & \frac{(q^{m_{1}}-1)(q-1)}{q+1} ~times,\\
\frac{-(q^{m}-1)(q-1)}{(q^{m_{1}}-1)(q^{m_{2}}-1)}\left(1+q^{\frac{m_{1}+m_{2}}{2}}\right) & \frac{q(q^{m_{1}}-1)(q-1)}{q+1} ~times.
\end{cases}
\end{eqnarray*}
\item If $q$ is odd and $m_{1}\equiv2\pmod4$, we have $m_{2}\equiv2\pmod4$ due to $\gcd(\frac{m_{2}}{2},q-1)=1$. Since $\frac{m_{2}}{2}$ is odd and $\frac{m_{1}}{2}$ is odd, by Lemma \ref{le2} we have
\end{itemize}
\begin{eqnarray*}
G(\lambda_{1}^{\frac{q^{m_{1}}-1}{q+1}j})G(\overline{\lambda}_{2}^{\frac{q^{m_{2}}-1}{q+1}j})=(-1)^{2j}q^{\frac{m_{1}+m_{2}}{2}}=q^{\frac{m_{1}+m_{2}}{2}}, 1\leq j\leq q.
\end{eqnarray*}
For $b=\alpha_{1}^{s_{1}}, c=\alpha_{1}^{s_{2}}, 0\leq s_{1}\leq q^{m_{1}}-2, 0\leq s_{2}\leq q-2$, we have
\begin{eqnarray*}
\Omega(b,c)&=&\frac{-(q^{m}-1)(q-1)}{(q^{m_{1}}-1)(q^{m_{2}}-1)}\left(1+q^{\frac{m_{1}+m_{2}}{2}}\sum_{j=1}^{q}\overline{\lambda}_{1}^{\frac{q^{m_{1}}-1}{q+1}j}(b)\right)\\
&=&\frac{-(q^{m}-1)(q-1)}{(q^{m_{1}}-1)(q^{m_{2}}-1)}\left(1+q^{\frac{m_{1}+m_{2}}{2}}\sum_{j=1}^{q}\zeta_{q+1}^{j(s_{2}-s_{1})}\lambda_{1}^{\frac{q^{m_{1}}-1}{q+1}j}(c)\right).
\end{eqnarray*}
For $(s_{2}-s_{1})\equiv0\pmod{q+1}$, we have
\begin{eqnarray*}
\sum_{j=1}^{q}\zeta_{q+1}^{j(s_{2}-s_{1})}=q
\end{eqnarray*}
and
\begin{eqnarray*}
\Omega(b,c)=\frac{-(q^{m}-1)(q-1)}{(q^{m_{1}}-1)(q^{m_{2}}-1)}\left(1+q^{\frac{m_{1}+m_{2}}{2}+1}\right).
\end{eqnarray*}
For $(s_{2}-s_{1})\equiv\frac{q+1}{2}\pmod{q+1}$, we have
\begin{eqnarray*}
\sum_{j=1}^{q}\zeta_{q+1}^{j(s_{2}-s_{1})}=-1
\end{eqnarray*}
and
\begin{eqnarray*}
\Omega(b,c)=\frac{-(q^{m}-1)(q-1)}{(q^{m_{1}}-1)(q^{m_{2}}-1)}\left(1-q^{\frac{m_{1}+m_{2}}{2}}\right).
\end{eqnarray*}
For $(s_{2}-s_{1})\not\equiv0,\frac{q+1}{2}\pmod{q+1}$, one can see that
\begin{eqnarray*}
\zeta_{q+1}^{s_{2}-s_{1}}+\zeta_{q+1}^{3(s_{2}-s_{1})}+\cdots+\zeta_{q+1}^{q(s_{2}-s_{1})}=0
\end{eqnarray*}
and
\begin{eqnarray*}
\zeta_{q+1}^{2(s_{2}-s_{1})}+\zeta_{q+1}^{4(s_{2}-s_{1})}+\cdots+\zeta_{q+1}^{(q-1)(s_{2}-s_{1})}=-1.
\end{eqnarray*}
This implies that
\begin{eqnarray*}
\sum_{j=1}^{q}\zeta_{q+1}^{j(s_{2}-s_{1})}=-1
\end{eqnarray*}
and
\begin{eqnarray*}
\Omega(b,c)=\frac{-(q^{m}-1)(q-1)}{(q^{m_{1}}-1)(q^{m_{2}}-1)}\left(1-q^{\frac{m_{1}+m_{2}}{2}}\right).
\end{eqnarray*}
Hence, the value distribution of $\Omega(b,c)$ is
\begin{eqnarray*}
\Omega(b,c)=
\begin{cases}
\frac{-(q^{m}-1)(q-1)}{(q^{m_{1}}-1)(q^{m_{2}}-1)}\left(1+q^{\frac{m_{1}+m_{2}}{2}+1}\right) & \frac{(q^{m_{1}}-1)(q-1)}{q+1} times,\\
\frac{-(q^{m}-1)(q-1)}{(q^{m_{1}}-1)(q^{m_{2}}-1)}\left(1-q^{\frac{m_{1}+m_{2}}{2}}\right) & \frac{q(q^{m_{1}}-1)(q-1)}{q+1} times.
\end{cases}
\end{eqnarray*}
Note that the value distribution of $\Omega(b,c)$ can be represented in a unified form for $e=2$. The proof is completed.
\end{proof}

\begin{lemma}\label{le5}
Let $l=2, e=1$ and other notations and hypothesises be the same as those of Lemma \ref{le3}. If $( b,c)$ runs through $\mathbb{F}_{q^{m_{1}}}^{*}\times\mathbb{F}_{q}^{*}$, then the value distribution of $\Omega(b,c), b\in\mathbb{F}_{q^{m_{1}}}^{*}, c\in \mathbb{F}_{q}^{*}$is given as follows.
\begin{eqnarray*}
\Omega(b,c)=
\begin{cases}
\frac{(q^{m}-1)(q-1)}{(q^{m_{1}}-1)(q^{m_{2}}-1)}\left(-1-q^{\frac{m_{1}+m_{2}+1}{2}}\right) & \frac{(q^{m_{1}}-1)(q-1)}{2} times,\\
\frac{(q^{m}-1)(q-1)}{(q^{m_{1}}-1)(q^{m_{2}}-1)}\left(-1+q^{\frac{m_{1}+m_{2}+1}{2}}\right) & \frac{(q^{m_{1}}-1)(q-1)}{2} times.
\end{cases}
\end{eqnarray*}
\end{lemma}
\begin{proof}
Since $l=2, e=1$, by Lemma \ref{le3} we have that
\begin{eqnarray*}
\Omega(b,c)=\frac{(q^{m}-1)(q-1)}{(q^{m_{1}}-1)(q^{m_{2}}-1)}\sum\limits_{s\in S}G(\lambda_{1}^{t_{1}s})G(\overline{\lambda}_{2}^{t_{2}s})\overline{\lambda}_{1}^{t_{1}s}(b)G(\overline{\lambda}^{m_{1}s})\lambda_{1}^{t_{1}s}(c),
\end{eqnarray*}
where $S=\left\{0, \frac{q-1}{2}\right\}$. It is clear that $m_{1}$ is odd and $m_{2}$ is even. Hence, by Lemma \ref{le4},
\begin{eqnarray*}
\Omega(b,c)&=&\frac{(q^{m}-1)(q-1)}{(q^{m_{1}}-1)(q^{m_{2}}-1)}\left(-1+G(\lambda_{1}^{\frac{q^{m_{1}}-1}{2}})G(\overline{\lambda}_{2}^{\frac{q^{m_{2}}-1}{2}})\overline{\lambda}_{1}^{\frac{q^{m_{1}}-1}{2}}(b)G(\overline{\lambda}^{\frac{q-1}{2}m_{1}})\right.\\
&\;&\left.\times\lambda_{1}^{\frac{q^{m_{1}}-1}{2}}(c)\right)\\
&=&\frac{(q^{m}-1)(q-1)}{(q^{m_{1}}-1)(q^{m_{2}}-1)}\left(-1+G(\lambda_{1}^{\frac{q^{m_{1}}-1}{2}})G(\overline{\lambda}_{2}^{\frac{q^{m_{2}}-1}{2}})\overline{\lambda}_{1}^{\frac{q^{m_{1}}-1}{2}}(b)G(\overline{\lambda}^{\frac{q-1}{2}})\right.\\
&\;&\left.\times\lambda_{1}^{\frac{q^{m_{1}}-1}{2}}(c)\right)\\
&=&\frac{(q^{m}-1)(q-1)}{(q^{m_{1}}-1)(q^{m_{2}}-1)}\left(-1-(-1)^{\frac{(p-1)(m_{1}+m_{2}+1)t}{4}}q^{\frac{m_{1}+m_{2}+1}{2}}\overline{\lambda}_{1}^{\frac{q^{m_{1}}-1}{2}}(b)\right.\\
&\;&\left.\lambda_{1}^{\frac{q^{m_{1}}-1}{2}}(c)\right)\\
&=&
\begin{cases}
\frac{(q^{m}-1)(q-1)}{(q^{m_{1}}-1)(q^{m_{2}}-1)}\left(-1-q^{\frac{m_{1}+m_{2}+1}{2}}\right) & \frac{(q^{m_{1}}-1)(q-1)}{2} times,\\
\frac{(q^{m}-1)(q-1)}{(q^{m_{1}}-1)(q^{m_{2}}-1)}\left(-1+q^{\frac{m_{1}+m_{2}+1}{2}}\right) & \frac{(q^{m_{1}}-1)(q-1)}{2} times,
\end{cases}
\end{eqnarray*}
where $( b,c)$ runs through $\mathbb{F}_{q^{m_{1}}}^{*}\times\mathbb{F}_{q}^{*}$,
\end{proof}
\section{The weight distribution and the locality of $\overline{\mathcal{C}_{D}}$}\label{sec4}

In this section, we will first  give the weight distribution of $\overline{\mathcal{C}_{D}}$ in some special cases. Then we will determine the locality of $\overline{\mathcal{C}_{D}}$.

\subsection{The weight distribution of $\overline{\mathcal{C}_{D}}$}
Since the norm function
\begin{eqnarray*}
\mathrm{N}_{q^{m}/q^{m_{2}}}:\mathbb{F}_{q^{m}}^*\rightarrow\mathbb{F}_{q^{m_{2}}}^*,x\mapsto y=x^{\frac{q^{m}-1}{q^{m_{2}}-1}},
\end{eqnarray*}
is an epimorphism of two multiplicative groups and the trace function $\mathrm{Tr}_{q^{m_{2}}/q}:\mathbb{F}_{q^{m_{2}}}\rightarrow\mathbb{F}_{q}$ is an epimorphism of two additive groups, the length of $\overline{\mathcal{C}_{D}}$ is equal to
\begin{eqnarray*}
n&=&\mid D \mid=\mid \text{Ker}(\mathrm{N}_{q^{m}/q^{m_{2}}})\mid \cdot(\mid \text{Ker}(\mathrm{Tr}_{q^{m_{2}}/q})\mid-1)+1\\
&=&\frac{(q^{m}-1)(q^{m_{2}}-q)}{q(q^{m_{2}}-1)}+1.
\end{eqnarray*}

In the following, we determine the weight distribution of $\overline{\mathcal{C}_{D}}$.

For each $b\in\mathbb{F}_{q^{m_{1}}}^{*}$, $c=0$, let
\begin{eqnarray*}
N_{b}=\sharp
\left\{\begin{array}{l}
x\in\mathbb{F}_{q^{m}}:\mathrm{Tr}_{q^{m_{1}}/q}(b\mathrm{N}_{q^{m}/q^{m_{1}}}(x))=0\\
~~~~~~~~~~~~~\text{and\quad$\mathrm{Tr}_{q^{m_{2}}/q}(\mathrm{N}_{q^{m}/q^{m_{2}}}(x))=0$}
\end{array}\right\}.
\end{eqnarray*}
From \cite[Section 4.1]{ZQ}, we have
\begin{eqnarray}\label{Nb}
N_{b}=\frac{q^{m}-1}{q^{2}}\left(1-\frac{q-1}{q^{m_{1}}-1}-\frac{q-1}{q^{m_{2}}-1}\right)+1+\frac{1}{q^{2}}\Delta(b).
\end{eqnarray}

For each $b\in\mathbb{F}_{q^{m_{1}}}^{*}$,$c\in\mathbb{F}_{q}^{*}$, denote by
\begin{eqnarray*}
N_{b,c}=\sharp
\left\{\begin{array}{l}
x\in\mathbb{F}_{q^{m}}:\mathrm{Tr}_{q^{m_{1}}/q}(b\mathrm{N}_{q^{m}/q^{m_{1}}}(x))+c=0\\
~~~~~~~~~~~~~\text{and\quad $\mathrm{Tr}_{q^{m_{2}}/q}(\mathrm{N}_{q^{m}/q^{m_{2}}}(x))=0$}
\end{array}\right\}.
\end{eqnarray*}
By the orthogonal relation of additive characters, for each $b\in\mathbb{F}_{q^{m_{1}}}^{*}$,$c\in\mathbb{F}_{q}^{*}$ we have
\begin{eqnarray*}
N_{b,c}&=&\frac{1}{q^{2}}\sum\limits_{x\in\mathbb{F}_{q^{m}}}
\sum\limits_{y\in\mathbb{F}_{q}}\zeta_{p}^{\mathrm{Tr}_{q/p}(y\mathrm{Tr}_{q^{m_{1}}/q}(b\mathrm{N}_{q^{m}/q^{m_{1}}}(x))+yc)}\\
&\;&\times \sum\limits_{z\in\mathbb{F}_{q}}\zeta_{p}^{\mathrm{Tr}_{q/p}(z\mathrm{Tr}_{q^{m_{2}}/q}(\mathrm{N}_{q^{m}/q^{m_{2}}}(x)))}\\
&=&\frac{1}{q^{2}}\sum\limits_{x\in\mathbb{F}_{q^{m}}}
\sum\limits_{y\in\mathbb{F}_{q}}\zeta_{p}^{\mathrm{Tr}_{q^{m_{1}}/p}(yb\mathrm{N}_{q^{m}/q^{m_{1}}}(x))+\mathrm{Tr}_{q/p}(yc)}
\sum\limits_{z\in\mathbb{F}_{q}}\zeta_{p}^{\mathrm{Tr}_{q^{m_{2}}/p}(z\mathrm{N}_{q^{m}/q^{m_{2}}}(x))}\\
&=&\frac{1}{q^{2}}\sum\limits_{x\in\mathbb{F}_{q^{m}}}\sum\limits_{y\in\mathbb{F}_{q}}\chi_{1}(yb\mathrm{N}_{q^{m}/q^{m_{1}}}(x))\chi(yc)
\sum\limits_{z\in\mathbb{F}_{q}}\chi_{2}(z\mathrm{N}_{q^{m}/q^{m_{2}}}(x))\\
&=&\frac{1}{q^{2}}+\frac{1}{q^{2}}\sum\limits_{x\in\mathbb{F}_{q^{m}}^{*}}1+\frac{1}{q^{2}}\sum\limits_{y\in\mathbb{F}_{q}^{*}}\chi(yc)
+\frac{1}{q^{2}}\sum\limits_{z\in\mathbb{F}_{q}^{*}}1\\
&\;&+\frac{1}{q^{2}}\sum\limits_{x\in\mathbb{F}_{q^{m}}^{*}}\sum\limits_{y\in\mathbb{F}_{q}^{*}}\chi_{1}(yb\mathrm{N}_{q^{m}/q^{m_{1}}}(x))\chi(yc)\\
&\;&+\frac{1}{q^{2}}\sum\limits_{x\in\mathbb{F}_{q^{m}}^{*}}\sum\limits_{z\in\mathbb{F}_{q}^{*}}\chi_{2}(z\mathrm{N}_{q^{m}/q^{m_{2}}}(x))
+\frac{1}{q^{2}}\sum\limits_{y\in\mathbb{F}_{q}^{*}}\sum\limits_{z\in\mathbb{F}_{q}^{*}}\chi(yc)\\
&\;&+\frac{1}{q^{2}}\sum\limits_{x\in\mathbb{F}_{q^{m}}^{*}}\sum\limits_{y,z\in\mathbb{F}_{q}^{*}}
\chi_{1}(yb\mathrm{N}_{q^{m}/q^{m_{1}}}(x))\chi(yc)\chi_{2}(z\mathrm{N}_{q^{m}/q^{m_{2}}}(x))\\
&=&\frac{q^{m}-1}{q^{2}}+\frac{1}{q^{2}}\sum\limits_{x\in\mathbb{F}_{q^{m}}^{*}}\sum\limits_{y\in\mathbb{F}_{q}^{*}}\chi_{1}(yb\mathrm{N}_{q^{m}/q^{m_{1}}}(x))\chi(yc)\\
&\;&+\frac{1}{q^{2}}\sum\limits_{x\in\mathbb{F}_{q^{m}}^{*}}\sum\limits_{z\in\mathbb{F}_{q}^{*}}\chi_{2}(z\mathrm{N}_{q^{m}/q^{m_{2}}}(x))+\frac{1}{q^{2}}\Omega(b,c).
\end{eqnarray*}
It is obvious that
\begin{eqnarray*}
\sum\limits_{x\in\mathbb{F}_{q^{m}}^{*}}\sum\limits_{y\in\mathbb{F}_{q}^{*}}\chi_{1}(yb\mathrm{N}_{q^{m}/q^{m_{1}}}(x))\chi(yc)
&=&\sum\limits_{x\in\mathbb{F}_{q^{m}}^{*}}\sum\limits_{y\in\mathbb{F}_{q}^{*}}\chi_{1}(ybx^{\frac{q^{m}-1}{q^{m_{1}}-1}})\chi(yc)\\
&=&\frac{q^{m}-1}{q^{m_{1}}-1}\sum\limits_{y\in\mathbb{F}_{q}^{*}}\chi(yc)\sum\limits_{x\in\mathbb{F}_{q^{m_{1}}}^{*}}\chi_{1}(ybx)\\
&=&\frac{q^{m}-1}{q^{m_{1}}-1},
\end{eqnarray*}
and
\begin{eqnarray*}
\sum\limits_{x\in\mathbb{F}_{q^{m}}^{*}}\sum\limits_{z\in\mathbb{F}_{q}^{*}}\chi_{2}(z\mathrm{N}_{q^{m}/q^{m_{2}}}(x))
&=&\sum\limits_{x\in\mathbb{F}_{q^{m}}^{*}}\sum\limits_{z\in\mathbb{F}_{q}^{*}}\chi_{2}(zx^{\frac{q^{m}-1}{q^{m_{2}}-1}})\\
&=&\frac{q^{m}-1}{q^{m_{2}}-1}\sum\limits_{z\in\mathbb{F}_{q}^{*}}\sum\limits_{x\in\mathbb{F}_{q^{m_{2}}}^{*}}\chi_{2}(zx)\\
&=&-(q-1)\frac{q^{m}-1}{q^{m_{2}}-1}.
\end{eqnarray*}
Then we have
\begin{eqnarray}\label{Nbc}
N_{b,c}=\frac{q^{m}-1}{q^{2}}\left(1+\frac{1}{q^{m_{1}}-1}-\frac{q-1}{q^{m_{2}}-1}\right)+\frac{1}{q^{2}}\Omega(b,c).
\end{eqnarray}

\begin{theorem}\label{thf}
Let $m$, $m_{1}$, $m_{2}$ be positive integers such that $m_{1}\mid m$, $m_{2}\mid m$, where $\gcd(m_{1},m_{2})=e=1$, $\gcd(\frac{m_{2}}{e},q-1)=l=1$. Then $\overline{\mathcal{C}_{D}}$ is a three-weight $\left[\frac{(q^{m}-1)(q^{m_{2}}-q)}{q(q^{m_{2}}-1)}+1,m_{1}+1\right]$ linear code and its weight distribution is given in Table 1. In particular, $\overline{\mathcal{C}_{D}}$ is a $q$-divisible self-orthogonal code if $m_{1}\geq2, m_{2}\geq2$ and $q$ is an odd prime power.

\begin{table}[h]
\begin{center}
\begin{minipage}{220pt}
\centering{\caption{Weight distribution of $\overline{\mathcal{C}_{D}}$ for $e=l=1$.}\label{tab1}}
\begin{tabular}{@{}cc@{}}
\toprule
Weight & Frequency  \\
\midrule
0    & 1  \\
$\frac{(q^{m}-1)(q^{m_{2}}-q)}{q(q^{m_{2}}-1)}+1$    & $q-1$\\
$\frac{q^{m_{1}-1}(q^{m}-1)(q-1)(q^{m_{2}-1}-1)}{(q^{m_{1}}-1)(q^{m_{2}}-1)}$    & $q^{m_{1}}-1$   \\
$\frac{(q^{m}-1)(q^{m_{2}-1}-1)(q^{m_{1}}-q^{m_{1}-1}-1)}{(q^{m_{1}}-1)(q^{m_{2}}-1)}+1$  & $(q^{m_{1}}-1)(q-1)$\\
\bottomrule
\end{tabular}
\end{minipage}
\end{center}
\end{table}
\end{theorem}
\begin{proof}
By Lemma \ref{db} and Lemma \ref{ob}, we know that
\begin{eqnarray*}
\Delta(b)=\frac{(q^{m}-1)(q-1)^{2}}{(q^{m_{1}}-1)(q^{m_{2}}-1)},
\end{eqnarray*}
and
\begin{eqnarray*}
\Omega(b,c)=\frac{-(q^{m}-1)(q-1)}{(q^{m_{1}}-1)(q^{m_{2}}-1)},
\end{eqnarray*}
where $\gcd(m_{1},m_{2})=e=1$, $\gcd(\frac{m_{2}}{e},q-1)=l=1$. For any codeword
$$\textbf{c}_{b,c}=\left(\mathrm{Tr}_{q^{m_{1}}/q}(b\mathrm{N}_{q^{m}/q^{m_{1}}}(x)\right)_{x\in D}+c\mathbf{1}\in \overline{\mathcal{C}_{D}},$$
let $\text{wt}(\textbf{c}_{b,c})$ denote its Hamming weight.\\
Case 1. If $b=0,c=0$, then we have
\begin{eqnarray*}
\text{wt}(\textbf{c}_{b,c})=0.
\end{eqnarray*}
Case 2. If $b=0,c\neq0$, then we have
\begin{eqnarray*}
\text{wt}(\textbf{c}_{b,c})=n=\frac{(q^{m}-1)(q^{m_{2}}-q)}{q(q^{m_{2}}-1)}+1.
\end{eqnarray*}
Case 3. If $b\neq0,c=0$, then by Equation (\ref{Nb}) we have
\begin{eqnarray*}
\text{wt}(\textbf{c}_{b,c})
&=&n-N_{b}\\
&=&\frac{q^{m_{1}-1}(q^{m}-1)(q-1)(q^{m_{2}-1}-1)}{(q^{m_{1}}-1)(q^{m_{2}}-1)}.
\end{eqnarray*}
Case 4. If $b\neq0,c\neq0$, then by Equation (\ref{Nbc}) we have
\begin{eqnarray*}
\text{wt}(\textbf{c}_{b,c})
&=&n-N_{b,c}\\
&=&\frac{(q^{m}-1)(q^{m_{2}-1}-1)(q^{m_{1}}-q^{m_{1}-1}-1)}{(q^{m_{1}}-1)(q^{m_{2}}-1)}+1.
\end{eqnarray*}
Then the weight distribution follows. The dimension of $\overline{\mathcal{C}_{D}}$ is $m_{1}+1$ as the zero codeword occurs once when $(b,c)$ runs through $\mathbb{F}_{q^{m_1}} \times \mathbb{F}_{q}$. In particular, $\overline{\mathcal{C}_{D}}$ is a $q$-divisible self-orthogonal code if $m_{1}\geq2, m_{2}\geq2$ by the weight distribution and Lemma \ref{th-selforthogonal}.
\end{proof}

\begin{example}
 Let $m=6$, $m_{1}=2$, $m_{2}=3$ and $q=3$. Then $\overline{\mathcal{C}_{D}}$ is a $3$-divisible self-orthogonal code with parameters $[225,3,141]$
and weight enumerator $1+16z^{141}+8z^{168}+2z^{225}$. It is confirmed to be correct by Magma.
\end{example}

\begin{theorem}\label{ths}
Let $m$, $m_{1}$, $m_{2}$ be positive integers such that $m_{1}\mid m$, $m_{2}\mid m$, where $\gcd(m_{1},m_{2})=e=2$, $\gcd(\frac{m_{2}}{e},q-1)=l=1$, $(m_{1},m_{2})\neq(2,2)$. Then $\overline{\mathcal{C}_{D}}$ is a five-weight linear code with parameters $\left[\frac{(q^{m}-1)(q^{m_{2}}-q)}{q(q^{m_{2}}-1)}+1,m_{1}+1\right]$ and its weight distribution is given in Table \ref{tab2}. If $m=m_{1}, m_{2}=2$ and $q>2$, the minimum distance of $\overline{\mathcal{C}_{D}}^{\perp}$ is $3$. If $m=m_{1}>4, m_{2}=2$ and $q=2$, then $\overline{\mathcal{C}_{D}}^{\perp}$ is distance-optimal according to the sphere-packing bound with minimum distance $4$.  Besides, $\overline{\mathcal{C}_{D}}$ is a $q$-divisible self-orthogonal code if $m_{1}\geq2, m_{2}\geq2, m_{1}+m_{2}\geq6$ and $q$ is an odd prime power.

\begin{table}[h]
\begin{center}
\begin{minipage}{330pt}
\centering{\caption{Weight distribution of $\overline{\mathcal{C}_{D}}$ for $e=2,l=1$.}\label{tab2}}%
\begin{tabular}{@{}cc@{}}
\toprule
Weight & Frequency  \\
\midrule
0    & 1  \\
$\frac{(q^{m}-1)(q^{m_{2}}-q)}{q(q^{m_{2}}-1)}+1$  & $q-1$  \\
$\frac{(q^{m}-1)(q-1)\left(q^{m_{1}-1}(q^{m_{2}-1}-1)-(-1)^{\frac{m_{1}+m_{2}}{2}}(q-1)q^{\frac{m_{1}+m_{2}}{2}-1}\right)}{(q^{m_{1}}-1)(q^{m_{2}}-1)}$    & $\frac{q^{m_{1}}-1}{q+1}$  \\
$\frac{(q^{m}-1)(q-1)\left(q^{m_{1}-1}(q^{m_{2}-1}-1)-(-1)^{\frac{m_{1}+m_{2}}{2}+1}(q-1)q^{\frac{m_{1}+m_{2}}{2}-2}\right)}{(q^{m_{1}}-1)(q^{m_{2}}-1)}$    & $\frac{q(q^{m_{1}}-1)}{q+1}$  \\
$\frac{(q^{m}-1)\left((q^{m_{2}-1}-1)(q^{m_{1}}-q^{m_{1}-1}-1)+(-1)^{\frac{m_{1}+m_{2}}{2}}(q-1)q^{\frac{m_{1}+m_{2}}{2}-1}\right)}{(q^{m_{1}}-1)(q^{m_{2}}-1)}+1$    & $\frac{(q^{m_{1}}-1)(q-1)}{q+1}$   \\
$\frac{(q^{m}-1)\left((q^{m_{2}-1}-1)(q^{m_{1}}-q^{m_{1}-1}-1)+(-1)^{\frac{m_{1}+m_{2}}{2}+1}(q-1)q^{\frac{m_{1}+m_{2}}{2}-2}\right)}{(q^{m_{1}}-1)(q^{m_{2}}-1)}+1$    & $\frac{q(q^{m_{1}}-1)(q-1)}{q+1}$   \\
\bottomrule
\end{tabular}
\end{minipage}
\end{center}
\end{table}
\end{theorem}
\begin{proof}
By Lemma \ref{db} and Lemma \ref{ob}, we know that
\begin{eqnarray*}
\Delta(b)=\frac{(q^{m}-1)(q-1)^{2}}{(q^{m_{1}}-1)(q^{m_{2}}-1)}
\end{eqnarray*}
and
\begin{eqnarray*}
\Omega(b,c)=
\begin{cases}
\frac{-(q^{m}-1)(q-1)}{(q^{m_{1}}-1)(q^{m_{2}}-1)}\left(1+(-1)^{\frac{m_{1}+m_{2}}{2}}q^{\frac{m_{1}+m_{2}}{2}+1}\right) & \frac{(q^{m_{1}}-1)(q-1)}{q+1} ~times,\\
\frac{-(q^{m}-1)(q-1)}{(q^{m_{1}}-1)(q^{m_{2}}-1)}\left(1+(-1)^{\frac{m_{1}+m_{2}}{2}+1}q^{\frac{m_{1}+m_{2}}{2}}\right) & \frac{q(q^{m_{1}}-1)(q-1)}{q+1} ~times,
\end{cases}
\end{eqnarray*}
where $\gcd(m_{1},m_{2})=e=2$, $\gcd(\frac{m_{2}}{e},q-1)=l=1$. For any codeword
$$\textbf{c}_{b,c}=\left(\mathrm{Tr}_{q^{m_{1}}/q}(b\mathrm{N}_{q^{m}/q^{m_{1}}}(x)\right)_{x\in D}+c\mathbf{1}\in \overline{\mathcal{C}_{D}},$$
let $\text{wt}(\textbf{c}_{b,c})$ denote its Hamming weight.\\
Case 1. If $b=0,c=0$, then we have
\begin{eqnarray*}
\text{wt}(\textbf{c}_{b,c})=0.
\end{eqnarray*}
Case 2. If $b=0,c\neq0$, then we have
\begin{eqnarray*}
\text{wt}(\textbf{c}_{b,c})=n=\frac{(q^{m}-1)(q^{m_{2}}-q)}{q(q^{m_{2}}-1)}+1.
\end{eqnarray*}
Case 3. If $b\neq0,c=0$, then by Equation (\ref{Nb}) we have
\begin{eqnarray*}
\text{wt}(\textbf{c}_{b,c})
&=&n-N_{b}\\
&=&
\left\{\begin{array}{ll}
\frac{(q^{m}-1)(q-1)\left(q^{m_{1}-1}(q^{m_{2}-1}-1)-(-1)^{\frac{m_{1}+m_{2}}{2}}(q-1)q^{\frac{m_{1}+m_{2}}{2}-1}\right)}{(q^{m_{1}}-1)(q^{m_{2}}-1)}\\ ~~~~~~~~~~~~~~~~~~~~~~~~~~~~~~~~~~~~~~~~~~~~~~~~~~~~~~~~~~~\text{$\frac{q^{m_{1}}-1}{q+1}$ times,}  \\
\frac{(q^{m}-1)(q-1)\left(q^{m_{1}-1}(q^{m_{2}-1}-1)-(-1)^{\frac{m_{1}+m_{2}}{2}+1}(q-1)q^{\frac{m_{1}+m_{2}}{2}-2}\right)}{(q^{m_{1}}-1)(q^{m_{2}}-1)}\\    ~~~~~~~~~~~~~~~~~~~~~~~~~~~~~~~~~~~~~~~~~~~~~~~~~~~~~~~~~~~\text{$\frac{q(q^{m_{1}}-1)}{q+1}$ times.}  \\
\end{array}\right.
\end{eqnarray*}
Case 4. If $b\neq0,c\neq0$, then by Equation (\ref{Nbc}) we have
\begin{eqnarray*}
\text{wt}(\textbf{c}_{b,c})
&=&n-N_{b,c}\\
&=&\left\{\begin{array}{ll}
\frac{(q^{m}-1)\left((q^{m_{2}-1}-1)(q^{m_{1}}-q^{m_{1}-1}-1)+(-1)^{\frac{m_{1}+m_{2}}{2}}(q-1)q^{\frac{m_{1}+m_{2}}{2}-1}\right)}{(q^{m_{1}}-1)(q^{m_{2}}-1)}+1\\    ~~~~~~~~~~~~~~~~~~~~~~~~~~~~~~~~~~~~~~~~~~~~~~~~~~~~~~~~~~~~~\text{$\frac{(q^{m_{1}}-1)(q-1)}{q+1}$ times,}   \\
\frac{(q^{m}-1)\left((q^{m_{2}-1}-1)(q^{m_{1}}-q^{m_{1}-1}-1)+(-1)^{\frac{m_{1}+m_{2}}{2}+1}(q-1)q^{\frac{m_{1}+m_{2}}{2}-2}\right)}{(q^{m_{1}}-1)(q^{m_{2}}-1)}+1\\   ~~~~~~~~~~~~~~~~~~~~~~~~~~~~~~~~~~~~~~~~~~~~~~~~~~~~~~~~~~~~~\text{$\frac{q(q^{m_{1}}-1)(q-1)}{q+1}$ times.}   \\
\end{array}\right.
\end{eqnarray*}
Then the weight distribution follows. The dimension of $\overline{\mathcal{C}_{D}}$ is $m_{1}+1$ as the zero codeword occurs once when $(b,c)$ runs through $\mathbb{F}_{q^{m_1}} \times \mathbb{F}_{q}$.

In the following, we will determine the minimum distance of $\overline{\mathcal{C}_{D}}^{\perp}$ with $m=m_{1}, m_{2}=2$. Denote
\begin{eqnarray*}
\omega_{1}=\frac{(q^{m}-1)(q^{m_{2}}-q)}{q(q^{m_{2}}-1)}+1,
\end{eqnarray*}
\begin{eqnarray*}
\omega_{2}=\frac{(q^{m}-1)(q-1)\left(q^{m_{1}-1}(q^{m_{2}-1}-1)-(-1)^{\frac{m_{1}+m_{2}}{2}}(q-1)q^{\frac{m_{1}+m_{2}}{2}-1}\right)}{(q^{m_{1}}-1)(q^{m_{2}}-1)},
\end{eqnarray*}
\begin{eqnarray*}
\omega_{3}=\frac{(q^{m}-1)(q-1)\left(q^{m_{1}-1}(q^{m_{2}-1}-1)-(-1)^{\frac{m_{1}+m_{2}}{2}+1}(q-1)q^{\frac{m_{1}+m_{2}}{2}-2}\right)}{(q^{m_{1}}-1)(q^{m_{2}}-1)},
\end{eqnarray*}
\begin{eqnarray*}
\omega_{4}=\frac{\left((q^{m_{2}-1}-1)(q^{m_{1}}-q^{m_{1}-1}-1)+(-1)^{\frac{m_{1}+m_{2}}{2}}(q-1)q^{\frac{m_{1}+m_{2}}{2}-1}\right)}{(q^{m_{1}}-1)(q^{m_{2}}-1)}\\
~~~~~~~~\cdot(q^{m}-1)+1,
\end{eqnarray*}
\begin{eqnarray*}
\omega_{5}=\frac{\left((q^{m_{2}-1}-1)(q^{m_{1}}-q^{m_{1}-1}-1)+(-1)^{\frac{m_{1}+m_{2}}{2}+1}(q-1)q^{\frac{m_{1}+m_{2}}{2}-2}\right)}{(q^{m_{1}}-1)(q^{m_{2}}-1)}\\
~~~~~~~~\cdot(q^{m}-1)+1.
\end{eqnarray*}
By the second,third, fourth and fifth Pless power moments in \cite{WV}, we have
\begin{eqnarray*}
\begin{cases}
\sum_{i=1}^{5}\omega_{i}A_{\omega_{i}}=q^{m}(qn-n-A_{1}^{\perp}),\\
\sum_{i=1}^{5}\omega_{i}^{2}A_{\omega_{i}}=q^{m-1}[(q-1)n(qn-n+1)-(2qn-q-2n+2)A_{1}^{\perp}+2A_{2}^{\perp}],\\
\sum_{i=1}^{5}\omega_{i}^{3}A_{\omega_{i}}=q^{m-2}[(q-1)n(q^{2}n^{2}-2qn^{2}+3qn-q+n^{2}-3n+2)\\
~~~~~~~~~~~~~~~~~~~-(3q^{2}n^{2}-3q^{2}n-6qn^{2}+12qn+q^{2}-6q+3n^{2}-9n+6)A_{1}^{\perp}\\
~~~~~~~~~~~~~~~~~~~+6(qn-q-n+2)A_{2}^{\perp}-6A_{3}^{\perp}],\\
\sum_{i=1}^{5}\omega_{i}^{4}A_{\omega_{i}}=q^{m-3}[(q-1)n(q^{3}n^{3}-3q^{2}n^{3}+6q^{2}n^{2}-4q^{2}n+q^{2}+3qn^{3}\\
~~~~~~~~~~~~~~~~~~~-12qn^{2}+15qn-6q-n^{3}+6n^{2}-11n+6)-(4q^{3}n^{3}\\
~~~~~~~~~~~~~~~~~~~-6q^{3}n^{2}+4q^{3}n-q^{3}-12q^{2}n^{3}+36q^{2}n^{2}-38q^{2}n+14q^{2}\\
~~~~~~~~~~~~~~~~~~~+12qn^{3}-54qn^{2}+78qn-36q-4n^{3}+24n^{2}-44n+24)A_{1}^{\perp}\\
~~~~~~~~~~~~~~~~~~~+(12q^{2}n^{2}-24q^{2}n+14q^{2}-24qn^{2}+84qn-72q+12n^{2}\\
~~~~~~~~~~~~~~~~~~~-60n+72)A_{2}^{\perp}-(24qn-36q-24n+72)A_{3}^{\perp}+24A_{4}^{\perp}],
\end{cases}
\end{eqnarray*}
where $n=\frac{(q^{m}-1)(q^{m_{2}}-q)}{q(q^{m_{2}}-1)}+1$, $m$ is even and $A_{\omega_{i}}(1\leq i\leq5)$ denotes the frequency of each weight $\omega_{i}$ in Table \ref{tab2}. Solving the above system of linear equations gives
\begin{eqnarray*}
A_{1}^{\perp}=A_{2}^{\perp}=0,
\end{eqnarray*}
\begin{eqnarray*}
A_{3}^{\perp}=\frac{(q^{m}-1)(q^{2}-3q+2)(q+q^{m}+2q^{2}+(-1)^{\frac{3m}{2}}q^{\frac{m}{2}+1}-(-1)^{\frac{3m}{2}}q^{\frac{m}{2}+2})}{6(q+1)^{3}}.
\end{eqnarray*}
If $q>2$, then $A_{3}^{\perp}>0$ and the minimum distance of $\overline{\mathcal{C}_{D}}^{\perp}$ is 3. If $q=2$, we can verify that $A_{3}^{\perp}=0$ and
\begin{eqnarray*}
A_{4}^{\perp}=\frac{(2^{m}-1)(2^{2m}-12\cdot2^{m}-16(-1)^{\frac{3m}{2}}2^{\frac{m}{2}})}{1944}>0,
\end{eqnarray*}
where $m=m_{1}>4$. Then the minimum distance of $\overline{\mathcal{C}_{D}}^{\perp}$ is $4$ for $q=2$ and $\overline{\mathcal{C}_{D}}^{\perp}$ is distance-optimal according to the sphere-packing bound.

The divisibility of $\overline{\mathcal{C}_{D}}$ follows from its weight distribution and its self-orthogonality follows from Lemma \ref{th-selforthogonal}.
\end{proof}

\begin{example}
 Let $m=m_{1}=4$, $m_{2}=2$, and $q=3$. Then $\overline{\mathcal{C}_{D}}$ is a $3$-divisible self-orthogonal $[21,5,12]$ ternary code with weight enumerator $1+100z^{12}+120z^{15}+20z^{18}+2z^{21}$, and its dual has parameters $[21,16,3]$ by Theorem \ref{ths}. These results are confirmed to be correct by Magma.
 Besides, we remark that $\overline{\mathcal{C}_{D}}$ is distance-optimal according to the Griesmer bound and its dual is also distance-optimal by the database at http://www.codetables.de/.
\end{example}

\begin{example}
 Let $m=m_{1}=6$, $m_{2}=2$, and $q=3$. Then $\overline{\mathcal{C}_{D}}$ is a $3$-divisible self-orthogonal $[183,7,108]$ ternary code with weight enumerator $1+182z^{108}+1092z^{120}+546z^{126}+364z^{129}+2z^{183}$, and its dual has parameters $[183,176,3]$ by Theorem \ref{ths}. These results are confirmed to be correct by Magma.
 Besides, we remark that the dual of $\overline{\mathcal{C}_{D}}$ is distance-optimal by the database at http://www.codetables.de/.
 \end{example}

\begin{example}
 Let $m=m_{1}=6$, $m_{2}=2$, and $q=2$. Then $\overline{\mathcal{C}_{D}}$ is a $[22,7,8]$ linear code with weight enumerator $1+21z^{8}+42z^{10}+42z^{12}+21z^{14}+z^{22}$, and its dual has parameters $[22,15,4]$ by Theorem \ref{ths}. These results are confirmed to be correct by Magma.
 Besides, we remark that $\overline{\mathcal{C}_{D}}$ and the dual of $\overline{\mathcal{C}_{D}}$ are both distance-optimal by the database at http://www.codetables.de/.
 \end{example}

\begin{example}
 Let $m=m_{1}=8$, $m_{2}=2$, and $q=2$. Then $\overline{\mathcal{C}_{D}}$ is a $[86,9,38]$ linear code with weight enumerator $1+85z^{38}+170z^{40}+170z^{46}+85z^{48}+z^{86}$, and its dual has parameters $[86,77,4]$ by Theorem \ref{ths}. These results are confirmed to be correct by Magma.
 Besides, we remark that the dual of $\overline{\mathcal{C}_{D}}$ is distance-optimal by the database at http://www.codetables.de/.
 \end{example}

\begin{theorem}\label{tht}
Let $m$, $m_{1}$, $m_{2}$ be positive integers such that $m_{1}\mid m$, $m_{2}\mid m$, where $\gcd(m_{1},m_{2})=e=1$, $\gcd(\frac{m_{2}}{e},q-1)=l=2$. Then $\overline{\mathcal{C}_{D}}$ is a four-weight linear code with parameters $[\frac{(q^{m}-1)(q^{m_{2}}-q)}{q(q^{m_{2}}-1)}+1,m_{1}+1]$ and its weight distribution is given in Table \ref{tab3}. In particular, $\overline{\mathcal{C}_{D}}$ is a $q$-divisible self-orthogonal code if $m_{1}\geq2, m_{2}\geq2, m_{1}+m_{2}\geq5$ and $q$ is an odd prime power.

\begin{table}[h]
\begin{center}
\begin{minipage}{280pt}
\centering{\caption{Weight distribution of $\overline{\mathcal{C}_{D}}$ for $e=1,l=2$.}\label{tab3}}%
\begin{tabular}{@{}cc@{}}
\toprule
Weight & Frequency  \\
\midrule
0    & 1  \\
$\frac{(q^{m}-1)(q^{m_{2}}-q)}{q(q^{m_{2}}-1)}+1$  & $q-1$  \\
$\frac{q^{m_{1}-1}(q^{m}-1)(q-1)(q^{m_{2}-1}-1)}{(q^{m_{1}}-1)(q^{m_{2}}-1)}$    & $q^{m_{1}}-1$   \\
$\frac{(q^{m}-1)\left((q^{m_{2}-1}-1)(q^{m_{1}}-q^{m_{1}-1}-1)+(q-1)q^{\frac{m_{1}+m_{2}-3}{2}}\right)}{(q^{m_{1}}-1)(q^{m_{2}}-1)}+1$    & $\frac{(q^{m_{1}}-1)(q-1)}{2}$   \\
$\frac{(q^{m}-1)\left((q^{m_{2}-1}-1)(q^{m_{1}}-q^{m_{1}-1}-1)-(q-1)q^{\frac{m_{1}+m_{2}-3}{2}}\right)}{(q^{m_{1}}-1)(q^{m_{2}}-1)}+1$    & $\frac{(q^{m_{1}}-1)(q-1)}{2}$   \\
\bottomrule
\end{tabular}
\end{minipage}
\end{center}
\end{table}
\end{theorem}
\begin{proof}
By Lemma \ref{db} and Lemma \ref{le5}, we know that
\begin{eqnarray*}
\Delta(b)=\frac{(q^{m}-1)(q-1)^{2}}{(q^{m_{1}}-1)(q^{m_{2}}-1)}
\end{eqnarray*}
and
\begin{eqnarray*}
\Omega(b,c)=
\begin{cases}
\frac{(q^{m}-1)(q-1)}{(q^{m_{1}}-1)(q^{m_{2}}-1)}\left(-1-q^{\frac{m_{1}+m_{2}+1}{2}}\right) & \frac{(q^{m_{1}}-1)(q-1)}{2} times,\\
\frac{(q^{m}-1)(q-1)}{(q^{m_{1}}-1)(q^{m_{2}}-1)}\left(-1+q^{\frac{m_{1}+m_{2}+1}{2}}\right) & \frac{(q^{m_{1}}-1)(q-1)}{2} times.
\end{cases}
\end{eqnarray*}
where $\gcd(m_{1},m_{2})=e=1$, $\gcd(\frac{m_{2}}{e},q-1)=l=2$. For any codeword
$$\textbf{c}_{b,c}=\left(\mathrm{Tr}_{q^{m_{1}}/q}(b\mathrm{N}_{q^{m}/q^{m_{1}}}(x)\right)_{x\in D}+c\mathbf{1}\in \overline{\mathcal{C}_{D}},$$
let $\text{wt}(\textbf{c}_{b,c})$ denote its Hamming weight. We consider the following cases.\\
Case 1. If $b=0,c=0$, then we have
\begin{eqnarray*}
\text{wt}(\textbf{c}_{b,c})=0.
\end{eqnarray*}
Case 2. If $b=0,c\neq0$, then we have
\begin{eqnarray*}
\text{wt}(\textbf{c}_{b,c})=n=\frac{(q^{m}-1)(q^{m_{2}}-q)}{q(q^{m_{2}}-1)}+1.
\end{eqnarray*}
Case 3. If $b\neq0,c=0$, then by Equation (\ref{Nb}) we have
\begin{eqnarray*}
\text{wt}(\textbf{c}_{b,c})
&=&n-N_{b}\\
&=&\frac{(q^{m}-1)(q^{m_{2}}-q)}{q(q^{m_{2}}-1)}+1.
\end{eqnarray*}
Case 4. If $b\neq0,c\neq0$, then by Equation (\ref{Nbc}) we have
\begin{eqnarray*}
\text{wt}(\textbf{c}_{b,c})
&=&n-N_{b,c}\\
&=&\left\{\begin{array}{ll}
\frac{(q^{m}-1)\left((q^{m_{2}-1}-1)(q^{m_{1}}-q^{m_{1}-1}-1)+(q-1)q^{\frac{m_{1}+m_{2}-3}{2}}\right)}{(q^{m_{1}}-1)(q^{m_{2}}-1)}+1\\    ~~~~~~~~~~~~~~~~~~~~~~~~~~~~~~~~~~~~~~~~~~~~~~\text{$\frac{(q^{m_{1}}-1)(q-1)}{2}$ times,}   \\
\frac{(q^{m}-1)\left((q^{m_{2}-1}-1)(q^{m_{1}}-q^{m_{1}-1}-1)-(q-1)q^{\frac{m_{1}+m_{2}-3}{2}}\right)}{(q^{m_{1}}-1)(q^{m_{2}}-1)}+1\\    ~~~~~~~~~~~~~~~~~~~~~~~~~~~~~~~~~~~~~~~~~~~~~~\text{$\frac{q(q^{m_{1}}-1)(q-1)}{2}$ times.}   \\
\end{array}\right.
\end{eqnarray*}
Then the weight distribution follows. The dimension of $\overline{\mathcal{C}_{D}}$ is $m_{1}+1$ as the zero codeword occurs once when $(b,c)$ runs through $\mathbb{F}_{q^{m_1}} \times \mathbb{F}_{q}$.
\end{proof}

\begin{example}
 Let $m=6$, $m_{1}=3$, $m_{2}=2$, and $q=3$. Then $\overline{\mathcal{C}_{D}}$ is a $3$-divisible self-orthogonal ternary $[183,4,99]$ code with enumerator
   $1+26z^{99}+26z^{126}+26z^{141}+2z^{183}$ by Theorem \ref{tht}. These results are confirmed to be correct by Magma.
  \end{example}

\subsection{The locality of $\overline{\mathcal{C}_{D}}$}

In this subsection, we will determine the locality of $\overline{\mathcal{C}_{D}}$. The conventional definition of linear locally recoverable codes is presented as follows.
\begin{definition}\cite{PEH}\label{deflr}
Let $\mathcal{C}$ be a linear code over $\mathbb{F}_{q}$ with a generator matrix $G=[\mathbf{g}_{1}, \mathbf{g}_{2},\cdots, \mathbf{g}_{n}]$. If $\mathbf{g}_{i}$
is a linear combination of at most $r$ other columns of $G$, then $r$ is called the locality of the $i$-th symbol
of each codeword of $\mathcal{C}$. Besides, $\mathcal{C}$ is called a locally recoverable code with locality $r$ if all the symbols
of codeword of $\mathcal{C}$ have locality $r$.
\end{definition}

In the following, we will use Definition \ref{deflr} to determine the locality of $\overline{\mathcal{C}_{D}}$.

\begin{lemma}\label{kx}
For the defining set $D$ in Equation (\ref{D}), if $x\in D$ and $l\in\mathbb{F}_{q}^{*}$, then we have $lx\in D$.
\end{lemma}
\begin{proof}
The desired conclusion directly follows from the properties of trance and norm functions.
\end{proof}

In the following, we determine the locality of $\overline{\mathcal{C}_{D}}$.

\begin{theorem}\label{locality}
Let $D$ be the defining set defined in Equation (\ref{D}). If $q>2$, then $\overline{\mathcal{C}_{D}}$ is a locally recoverable code with locality 2.
\end{theorem}
\begin{proof}
Let $\mathbb{F}_{q^{m_{1}}}^{*}=\langle\alpha_{1}\rangle$. Then $\{\alpha_{1}^{0}, \alpha_{1}^{1}, \cdots, \alpha_{1}^{m_{1}-1}\}$ is a $\mathbb{F}_{q}$-basis of $\mathbb{F}_{q^{m_{1}}}$. Let $d_{1}, d_{2}, \cdots, d_{n}$ be all the elements in $D$. Without loss of generality, let $d_{n}=0$ due to $0\in D$. By definition, the generator matrix $G$ of $\overline{\mathcal{C}_{D}}$ is given by
\begin{eqnarray*}
G:=
\begin{bmatrix}
1 & \cdots & 1\\
\mathrm{Tr}_{q^{m_{1}}/q}(\alpha_{1}^{0}\mathrm{N}_{q^{m}/q^{m_{1}}}(d_{1})) & \ldots & \mathrm{Tr}_{q^{m_{1}}/q}(\alpha_{1}^{0}\mathrm{N}_{q^{m}/q^{m_{1}}}(d_{n}))\\
\vdots & \ddots & \vdots\\
\mathrm{Tr}_{q^{m_{1}}/q}(\alpha_{1}^{m_{1}-1}\mathrm{N}_{q^{m}/q^{m_{1}}}(d_{1})) & \cdots & \mathrm{Tr}_{q^{m_{1}}/q}(\alpha_{1}^{m_{1}-1}\mathrm{N}_{q^{m}/q^{m_{1}}}(d_{n}))
\end{bmatrix}.
\end{eqnarray*}
For convenience, we assume that $\mathbf{g}_{i}$ denotes the $i$-th column of $G$, i.e., $$\mathbf{g}_{i}=\left(1, \mathrm{Tr}_{q^{m_{1}}/q}(\alpha_{1}^{0}\mathrm{N}_{q^{m}/q^{m_{1}}}(d_{i})), \cdots, \mathrm{Tr}_{q^{m_{1}}/q}(\alpha_{1}^{m_{1}-1}\mathrm{N}_{q^{m}/q^{m_{1}}}(d_{i}))\right)^{T},$$ where $1\leq i\leq n$. By Lemma \ref{kx}, we know that $ld_{i}\in D$ for any $l\in\mathbb{F}_{q}^{*}$ if $d_{i}\in D$. Let $\mathbb{F}_{q}^{*}=\langle\beta\rangle$ and $u=\beta^{-t} \in \mathbb{F}_{q}\backslash \{0,1\}$ such that $\gcd{(\frac{m}{m_{1}},q-1)}\mid t$ for $0<t<q-1$. Then there exists integer $s$ such that $\frac{sm}{m_{1}}\equiv t\pmod {q-1}$.
Let $l:=\beta^s$ and $v=1-u$. Then $u\cdot l^{\frac{m}{m_{1}}}=1$. For any $d_i$ with $1\leq i\leq n-1$ and $d_n=0$, let $d_{j}:=ld_{i}\in D$ and we can verify that
\begin{eqnarray*}
\begin{cases}
1&=u+v,\\
\mathrm{Tr}_{q^{m_{1}}/q}(\alpha_{1}^{0}\mathrm{N}_{q^{m}/q^{m_{1}}}(d_{i}))&=u\mathrm{Tr}_{q^{m_{1}}/q}(\alpha_{1}^{0}\mathrm{N}_{q^{m}/q^{m_{1}}}(d_{j}))\\
&~~+v\mathrm{Tr}_{q^{m_{1}}/q}(\alpha_{1}^{0}\mathrm{N}_{q^{m}/q^{m_{1}}}(d_{n})),\\
&\vdots\\
\mathrm{Tr}_{q^{m_{1}}/q}(\alpha_{1}^{m_{1}-1}\mathrm{N}_{q^{m}/q^{m_{1}}}(d_{i}))&=u\mathrm{Tr}_{q^{m_{1}}/q}(\alpha_{1}^{m_{1}-1}\mathrm{N}_{q^{m}/q^{m_{1}}}(d_{j}))\\
&~~+v\mathrm{Tr}_{q^{m_{1}}/q}(\alpha_{1}^{m_{1}-1}\mathrm{N}_{q^{m}/q^{m_{1}}}(d_{n})).
\end{cases}
\end{eqnarray*}
 This means that $\mathbf{g}_{i}$ is a linear combination of $\mathbf{g}_{j}$ and $\mathbf{g}_{n}$, where $1\leq i\leq n-1$. Since $v\neq0, 1$, it is clear that $\mathbf{g}_{n}$ is also a linear combination of $\mathbf{g}_{i}$ and $\mathbf{g}_{j}$. Then $\overline{\mathcal{C}_{D}}$ is a locally recoverable code with locality 2 according to Definition \ref{deflr}.
\end{proof}

In the following, we will present some optimal or almost optimal locally recoverable codes.
According to Inequality  (\ref{CMe}), for an $[n,k,d]$ locally recoverable code over $\mathbb{F}_{q}$ with locality $r$, we have
\begin{eqnarray}\label{bound-k}
k\leq\min\limits_{1\leq t\leq n/(r+1)}\left\{tr+k_{\text{opt}}^{(q)}(n-t(r+1),d)\right\}\leq r+k_{\text{opt}}^{(q)}(n-(r+1),d).
\end{eqnarray}

\begin{example}
Let $m=m_{1}=4$, $m_{2}=2$ and $q=3$. Then $\overline{\mathcal{C}_{D}}$ has parameters $[21,5,12]$ by Theorem \ref{ths} and locality 2 by Theorem \ref{locality}.
For $n=21,d=12,r=2$, by the Griesmer bound, we have
\begin{eqnarray*}
k_{\text{opt}}^{(q)}(n-(r+1),d)=k_{\text{opt}}^{(q)}(18,12)=3,
\end{eqnarray*}
Then the right hand of Inequality (\ref{bound-k}) equals $5$ and
$$\min\limits_{1\leq t\leq n/(r+1)}\left\{tr+k_{\text{opt}}^{(q)}(n-t(r+1),d)\right\}=5.$$
Therefore, $\overline{\mathcal{C}_{D}}$ is an optimal ternary $[21,5,12]$ locally recoverable code with locality $2$ with respect to the Cadambe-Mazumdar bound.
\end{example}

\begin{example}
Let $m=m_{1}=4$, $m_{2}=2$ and $q=4$. Then $\overline{\mathcal{C}_{D}}$ has parameters $[52,5,36]$ by Theorem \ref{ths} and locality 2 by Theorem \ref{locality}. Note that $\overline{\mathcal{C}_{D}}$ is a distance-optimal linear code according to the database at http://www.codetables.de/.
For $n=52$, $d=36$ and $r=2$, by the Griesmer bound, we have
\begin{eqnarray*}
k_{\text{opt}}^{(q)}(n-(r+1),d)=k_{\text{opt}}^{(q)}(49,36)=4.
\end{eqnarray*}
and
$$k_{\text{opt}}^{(q)}(n-2(r+1),d)=k_{\text{opt}}^{(q)}(46,36)=2.$$
Since $tr=2t\geq 6$ for $t\geq 3$, for the right hand of the Cadambe-Mazumdar bound we have
$$\min\limits_{1\leq t\leq n/(r+1)}\left\{tr+k_{\text{opt}}^{(q)}(n-t(r+1),d)\right\}=\min\limits_{t=1,2}\left\{tr+k_{\text{opt}}^{(q)}(n-t(r+1),d)\right\}=6.$$
Therefore, $\overline{\mathcal{C}_{D}}$ is an almost optimal quaternary $[52,5,36]$ locally recoverable code with locality $2$ with respect to the Cadambe-Mazumdar bound.
\end{example}

\section{Concluding remarks}\label{sec5}
In this paper, based on the augmentation technique, we constructed a family of linear codes by the trace and norm functions.
The parameters and weight distribution of the codes were determined in some cases by Gaussian sums.
The codes have only three, four or five nonzero weights in these cases. In particular, the codes constructed in this paper have the following interesting properties.
\begin{enumerate}
\item[$\bullet$] The codes in Theorems \ref{thf}, \ref{ths} and \ref{tht} are all $q$-divisible codes for large enough $m_1$ and $m_2$;
\item[$\bullet$] The codes in Theorems \ref{thf}, \ref{ths} and \ref{tht} are all self-orthogonal codes for odd prime power $q$ and large enough $m_1$ and $m_2$;
\item[$\bullet$] The codes in Theorem \ref{ths} contain two infinite subfamilies of projective codes. In Theorem \ref{ths}, there is also an infinite family of distance-optimal binary linear codes with parameters $\left[\frac{2^m+2}{3},\frac{2^m+2}{3}-m-1,4\right]$ for even $m>4$.
 \item[$\bullet$] If $q>2$, $\overline{\mathcal{C}_{D}}$ is a locally recoverable code with locality 2 by Theorem \ref{locality}. Some optimal or almost optimal locally recoverable code were also obtained.
\end{enumerate}
We remark that locally recoverable codes with locality $r\ll k$ are useful in  distributed storage systems.
The self-orthogonal codes in this paper can be used to construct even lattices by the main results in \cite{Z}.

\section*{}
\textbf{Availability of data and materials} Not applicable.

 \section*{Declarations}
\textbf{Conflict of interest} The authors declare that they have no conflicts of interest relevant to the content of this article.





\end{document}